\documentclass{amsart}
\usepackage[ruled,vlined]{algorithm2e}
\usepackage{wrapfig}
\usepackage{amsthm}
\usepackage{graphicx}
\usepackage[usenames,dvipsnames]{xcolor}
\usepackage{caption}
\usepackage{subcaption}
\usepackage{picins}

\newtheorem{theorem}{Theorem}
\newtheorem{lemma}{Lemma}
\newtheorem{proposition}{Proposition}

\newcommand{\norm}[1]{\left\Vert#1\right\Vert}

\newcommand{\set}[1]{\left\{#1\right\}}
\newcommand{\parr}[1]{\left (#1\right )}
\newcommand{\brac}[1]{\left [#1\right ]}

\newcommand{\Real}{\mathbb R}

\newcommand{\Integer}{\mathbb Z}

\newcommand{\too}{\rightarrow}

\newcommand{\A}{\mathcal{A}}

\newcommand{\area}[1]{\mathrm{Area}(#1)}
\newcommand{\aff}{\mathrm{aff}}
\newcommand{\sign}{\mathrm{sign}}
\newcommand{\tr}{\mathrm{tr}}

\def \M{\mathrm{\textbf{M}}} 
\def \K{\mathrm{\textbf{K}}} 
\def \L{\mathrm{\textbf{L}}} 
\def \C{\mathcal{C}}   
\def \OOmega{\mbox{\boldmath$\Omega$}}
\def \sphere{\mathbb{S}} 
\def \vecx{\mathrm{\textbf{x}}} 
\def \v{\mbox{\boldmath$v$}} 
\def \u{\mbox{\boldmath$u$}} 

\newcommand {\closure}[1]{\textrm{Closure}(#1)} 
\newcommand {\interior}[1]{\textrm{Interior}(#1)} 

\usepackage[normalem]{ulem}
\usepackage{amsmath}

\begin{document}

\title[]{Bijective Mappings Of Meshes With Boundary \\ And The Degree In Mesh Processing }
\author{Y. Lipman}
\address{Weizmann Institute Of Science}%
\email{Yaron.Lipman@weizmann.ac.il}%

\date{\today}
\maketitle

\begin{abstract}
This paper introduces three sets of sufficient conditions, for generating bijective simplicial mappings of manifold meshes.

A necessary condition for a simplicial mapping of a mesh to be injective is that it either maintains the orientation of all elements or flips all the elements. However, these conditions are known to be insufficient for injectivity of a simplicial map. In this paper we provide additional simple conditions that, together with the above mentioned necessary conditions guarantee injectivity of the simplicial map.

The first set of conditions generalizes classical global inversion theorems to the mesh (piecewise-linear) case. That is, proves that in case the boundary simplicial map is bijective and the necessary condition holds then the map is injective and onto the target domain.

The second set of conditions is concerned with mapping of a mesh to a polytope and replaces the (often hard) requirement of a bijective boundary map with a collection of linear constraints and guarantees that the resulting map is injective over the interior of the mesh and onto. These linear conditions provide a practical tool for optimizing a map of the mesh onto a given polytope while allowing the boundary map to adjust freely and keeping the injectivity property in the interior of the mesh. Allowing more freedom in the boundary conditions is useful for two reasons: a) it circumvents the hard task of providing a bijective boundary map, and b) it allows optimizing the boundary map \emph{together} with the simplicial map to achieve lower energy levels.

The third set of conditions adds to the second set the requirement that the boundary maps are orientation preserving as-well (with a proper definition of boundary map orientation). This set of conditions guarantees that the map is injective on the boundary of the mesh as-well as its interior. Several experiments using the sufficient conditions are shown for mapping triangular meshes.

A secondary goal of this paper is to advocate and develop the tool of degree in the context of mesh processing.

\end{abstract}


\section{Introduction}
Triangular and tetrahedral meshes are prominent in representing surfaces and volumes in various fields such as computer graphics and vision, medical imaging, and engineering. Many of the algorithms and applications that use mappings of meshes require injectivity of the map to operate correctly. Nevertheless, injectivity in general is a hard constraint and poses a real challenge to guarantee.

The main goal of this paper is to provide practical sufficient conditions that assures that a simplicial mapping $\Phi:\M\too\Real^d$ is an injection taking a $d$-dimensional compact mesh $\M$ onto a prescribed polytope $\Omega\subset \Real^d$.
%
A secondary goal is to advocate and develop the tool of degree in the context of mesh processing.

In classical analysis and elasticity theory global inversion theorems provide sufficient conditions for injectivity, started with the work of Hadamard, Darboux and Stoilow and followed with the work of many others, see for example \cite{MeistersOlech63,mcauley66,Ball81,Ciarlet87,Massey92,Giuseppe94,Cristea00}.
 Nevertheless, most of the previous results deal with smooth mappings and/or require local injectivity of the map, both unnatural for the mesh case (note that even if every element of a mesh is not inverted it is still not necessarily locally injective). More importantly, the sufficient conditions offered in previous works mostly assume that the mapping under consideration is bijective when restricted to the boundary of the domain. In practice, this condition is often too restrictive for two main (practical) reasons: first, constructing a bijective boundary mapping can be as challenging as the original problem, and second, we are often required in applications to find a mapping of a mesh that is minimizing a certain cost (or energy) and we do not know a-priori what would be the optimal boundary map.

A necessary condition for a simplicial map $\Phi$ to be injective is that the map either maintains the orientations of all elements or inverts the orientation of all elements. Nevertheless, as is known (and shown later in this paper) these conditions are not sufficient for injectivity. In this work we will prove three sets of conditions that in addition to the necessary conditions form sufficient conditions for injectivity.

The first set of conditions generalizes the global inversion theorems directly to meshes (albeit without the usual smoothness or local injectivity requirements). It shows that a sufficient condition for injectivity of a simplicial map $\Phi$ is that the boundary map $\Phi\vert_{\partial\M}:\partial\M\too\partial\OOmega$ is a bijection, together with the necessary conditions of orientation consistency.

The second set of conditions is aimed at the problem of mapping the mesh $\M$ onto a polytope $\OOmega$ and allows weaker boundary conditions where it is not required to provide a particular boundary mapping.
The benefit in this second set is that it only adds a set of \emph{linear} constraints (in addition to the necessary conditions) and therefore allows building algorithms that optimize over a \emph{collection} of boundary maps while guaranteeing that the map $\Phi$ is injective over the interior of the mesh $\M$ and that it covers the target domain $\Phi(\M)=\OOmega$. Note that it does not in general guarantee the injectivity of the boundary mapping as detailed later.

The third set of conditions is also aimed at mapping $M$ to a polytope $\OOmega$ and adds to the second set of conditions the requirement that the boundary maps are also orientation preserving (to be defined precisely soon) and by that guarantees the injectivity of the map over the entire $\M$ (including the boudnary of $\M$).

A related work to ours is the work of Floater \cite{Floater03one-to-onepiecewise} generalizing Tutte's drawing a planar graph paper \cite{Tutte:1963} that also provides a set of sufficient conditions for generating injective mappings of disk-type triangular meshes ($d=2$ meshes) mapped onto a convex polygonal domain. More general sufficient conditions were studied in \cite{Gortler06discreteone-forms}. Unfortunately, these constructions do not generalize to higher dimensions and/or non-convex target domains. Other related works that studied local and global injective mappings include \cite{Aronov:1993,Xu:2011:ETG:2286660.2286720,Schueller:LIM:2013,Schneider:2013:BCM} however these works do not seem to overlap strongly with our goals. The current paper also provides full background and mathematical underpinning to the injectivity arguments from our previous papers \cite{Lipman:2012:BDM,Aigerman:bd3d:2013}.

\section{Preliminaries and Main Results}
\label{s:preliminaries}

Our object of interest is a $d$-dimensional compact manifold mesh $\M=(\K)$, where $\K=\set{\sigma}$ is a $d$-dimensional finite simplicial complex, and $\sigma$ denotes a simplex (we will also use other greek letter such as $\tau,\kappa,\alpha$ to denote simplices of $\M$).
A \emph{face} $\sigma\in\K$ will be our generic name to any simplex in $\M$.
We will use the term $\ell$-face to denote a simplex of dimension exactly $\ell$. For example, $2$-face is a triangle, $3$-face is a tetrahedron (tet).  Each $\sigma\in\K$ has a fixed orientation:  all $\sigma\in\K_d$ are assigned with consistent orientation, $\sigma\in\partial\K_{d-1}$ are assigned with the induced orientation, and all other faces are given some arbitrary but fixed orientation. We will use the notation $\K_\ell$ to denote the subset of $\ell$-faces in $\K$, and so $\K=\K_0\cup\K_1\cup...\cup\K_d$. %
For $d=2$ (i.e., triangular mesh) we have  $\K=\K_0\cup\K_1\cup\K_2$, where $\K_0\subset \Real^n$, $n\geq d$, is a collection of points, $\K_1$ the set of edges, and $\K_2$ the collection of triangles. For $d=3$ (i.e., tetrahedral mesh), we have $\K=\K_0\cup\K_1\cup\K_2\cup\K_3$, where $\K_3$ is the set of tetrahedra.
Our simplices are considered by default closed sets (e.g., a point, a closed line segment, a triangle with its boundary). In this context the mesh $\M$ can be seen as the closed set of points in $\Real^n$ which is constructed as the union of all the simplices in $\K$. We will restrict our attention to compact, orientable, connected meshes $\M$ with boundary $\partial\M$.
We will mark by $\partial\M=(\partial \K)$ the boundary mesh of $\M$ (e.g., a polygon for $d=2$, and triangular mesh for $d=3$), and by $\partial\K$ the set of all faces that are contained in $\partial\M$. $\partial\K_\ell$ will denote the subset of boundary $\ell$-faces of $\M$, that is all $\ell$-faces $\sigma\in\K_\ell$ such that $\sigma\subset\partial\M$.


A simplicial map $\Phi:\M\too \Real^d$ is a  continuous map that is an affine map when restricted to each face $\sigma$ of $\M$, we denote this affine map as $\Phi\vert_{\sigma}$. A simplicial map is uniquely determined by setting the image position of each point ($0$-face), that is $u=\Phi(v)\in \Real^d$, for all $v\in\K_0$, and extending linearly over all faces.

Our target domain $\OOmega\subset\Real^d$ will be a (closed) polytope. Its boundary is denoted $\partial\OOmega$, and similarly to meshes, $\partial\OOmega_\ell$ will denote the collection of all boundary $\ell$-faces (i.e., polytope's faces of dimension $\ell$ contained in the boundary $\partial\OOmega$). For example, for $d=2$, $\partial\OOmega$ is a polygonal line, while for $d=3$, $\partial\OOmega$ is a polyhedral surface.

We are seeking conditions that guarantee that a simplicial mapping $\Phi:\M\too\OOmega$ is injective and onto. Let us start with a simple \emph{necessary} condition: if $\Phi$ is injective then it does not degenerate any $d$-face and either maintains the orientation of all $d$-faces or flips all orientations of  $d$-faces:
\begin{proposition}\label{prop:necessary}(Necessary condition for injectivity of a simplicial map)\\
An injective simplicial map $\Phi:\M\too\Real^d$ of a  $d$-dimensional compact mesh does not degenerate any $d$-face and satisfies exactly one of the following: 1) maintains the orientation of all $d$-faces, or 2) invert the orientation of all $d$-faces.
\end{proposition}
Indeed, if the affine map $\Phi\vert_\sigma$ is degenerate for some $\sigma\in\K_d$ then it is clearly not injective. So all $d$-faces are necessarily mapped to $d$-faces.  Next, assume (in negation) that some of the $d$-faces' orientation is preserved and some are inverted, then by connectivity there has to be two adjacent (i.e., sharing a $d-1$ face) $d$-faces $\sigma_1,\sigma_2\in\K_d$ such that $\Phi(\sigma_1)$ and $\Phi(\sigma_2)$ have opposite orientations. Since the two simplices $\Phi(\sigma_1),\Phi(\sigma_2)$ are sharing a $d-1$-face the intersection of their interiors is not empty leading to a contradiction with the assumption that $\Phi$ is injective.

\begin{wrapfigure}[14]{r}{0.35\textwidth}
  \begin{center}\vspace{-0.0cm}\hspace{-0.3cm}
    \includegraphics[width=0.35\textwidth]{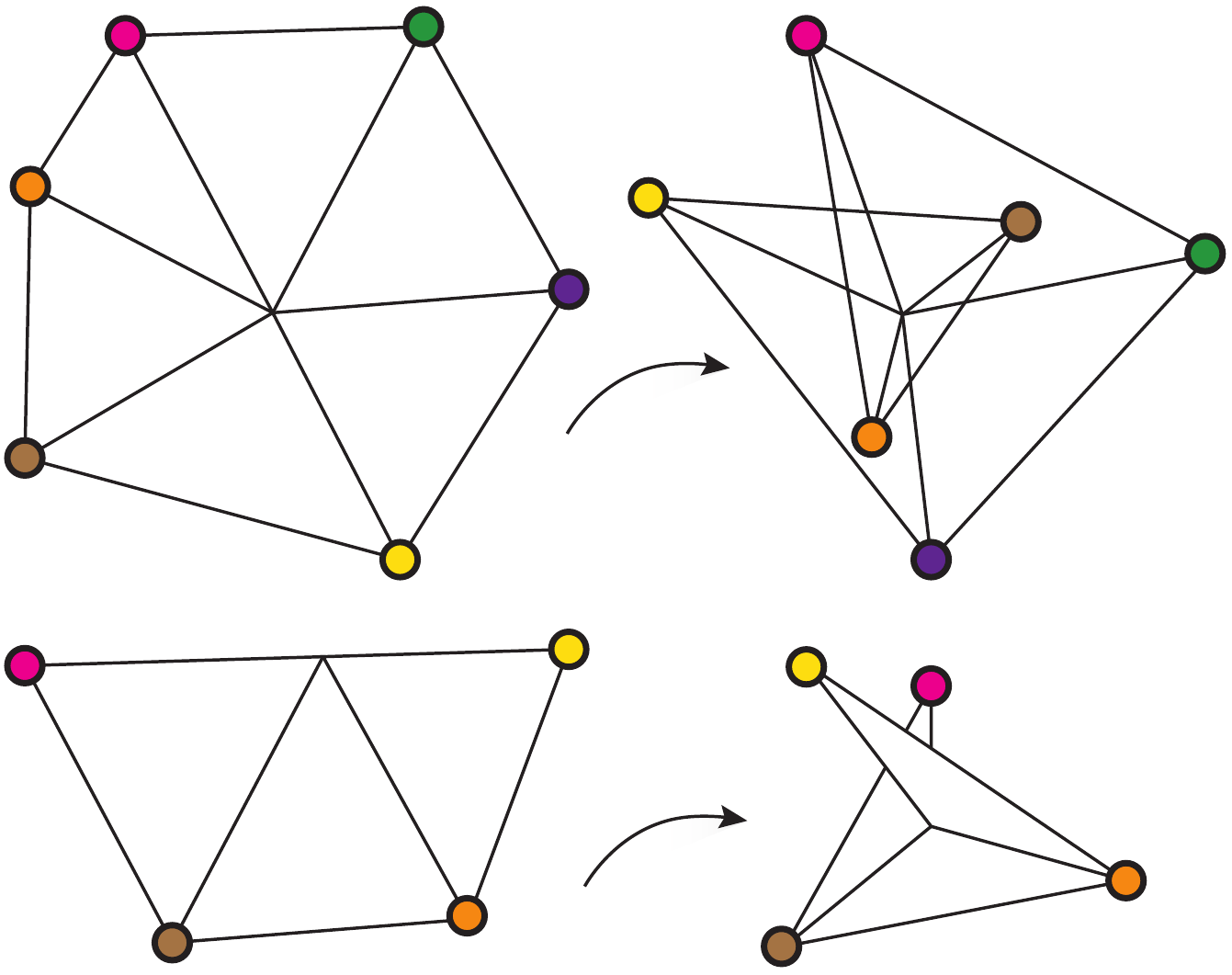}\vspace{-0.4cm}
  \end{center}
  \caption{Necessary condition for injectivity is not sufficient.}\label{fig:counter_example}\vspace{-0.6cm}
\end{wrapfigure}
A simplicial map that satisfies the above necessary condition will be said to have the \emph{consistent orientation} property. In this paper we will restrict our attention to \emph{orientation preserving} simplicial maps. That is, all the affine maps $\Phi\vert_\sigma$, $\sigma\in\K_d$ are orientation preserving. The orientation reversing maps can be treated similarly.

Maintaining the orientation of the $d$-faces is not sufficient to guarantee injectivity; Figure \ref{fig:counter_example} shows two counter-example for $d=2$ (i.e., a triangular mesh). Similar and more elaborate examples can be given for the $d=3$ case (tetrahedral mesh). Our goal is to characterize additional ``simple'' conditions that will assure, along with the consistent orientation condition, injectivity of the simplicial map.

Before getting to the main results we will prove the following theorem, different versions of which appeared in \cite{Lipman:2012:BDM,Aigerman:bd3d:2013}. The current paper provides mathematical underpinning and generalizations to the arguments provided in these earlier works.
\begin{theorem}\label{thm:bijectivity1}
A non-degenerate orientation preserving simplicial map $\Phi:\M\too\Real^d$ of a $d$-dimensional compact mesh with boundary $\M$ is a bijection $\Phi:\M\too\OOmega$ if the boundary map $\Phi\vert_{\partial\M}:\partial\M\too \partial\OOmega$ is bijective.
\end{theorem}
In practice, using this theorem to build an injective simplicial map $\Phi:\M\too\OOmega$ requires to a-priori come up with a bijective boundary map $\Phi\vert_{\partial \M}$, and only later search for a simplicial map satisfying this boundary condition. In many common scenarios it is better to search for a map without restricting the boundary map to a fixed boundary map. For example, if seeking low distortion mappings, like is done in the parameterization problem, allowing the vertices at the boundary to move would help reaching a lower distortion level (or even an existence of a solution). Therefore, we will relax the above sufficient conditions and provide more flexible setting to guarantee that a map is injective. These conditions, that are described next, will allow certain freedom in the boundary maps while keeping the injectivity property. We start with some preparations.

As mentioned above we would like to avoid prescribing a specific boundary map $\Phi\vert_{\partial\M}$ and allow each boundary face $\sigma\in\partial\K$ that is mapped onto $\partial\OOmega$ to ``slide'' on its target face of the polytope. For that end, every boundary $d-1$-face $\sigma\in\partial\K_{d-1}$ is assigned with a target boundary face of the polytope of the same dimension $\tau\in\partial\OOmega_{d-1}$. We will denote such an assignment by a function
\begin{equation}\label{e:assignment}
\A:\partial\K_{d-1}\too\partial\OOmega_{d-1}.
\end{equation}
The idea is that every face $\sigma\in\partial\K_{d-1}$ is mapped to somewhere on the planar boundary face $\tau=\A(\sigma)\in\partial\OOmega_{d-1}$ but we don't fix a-priori where exactly.

Unfortunately, although this condition is  more flexible than fixing the boundary map it can lead to rather complicated non-convex constraints. For example, if one of the boundary faces of the polytope is non-convex this constrain will turn out non-convex as-well. See for example the polytope depicted in Figure \ref{fig:3dpolytope}.  Hence, we will suggest a relaxation of this condition: we ask all $d-1$-faces $\sigma\in\partial\K_{d-1}$ to satisfy \begin{equation}\label{e:linear_con}
\Phi(\sigma)\subset\aff(\A(\sigma)),
\end{equation}
where $\aff(B)$ means the affine closure of a set $B$, namely the smallest affine set containing $B$. In particular, these are \emph{linear} equations when formulated in the unknowns of $\Phi$ (i.e., the target location of each vertex) and hence useful in practice. As noted above, in contrast to Eq.~(\ref{e:linear_con}), the condition $\Phi(\sigma)\subset \A(\sigma)$ is not convex in the case the face $\A(\sigma)$ is not convex, which is often the case.
Even in the case the face $\A(\sigma)$ is a convex face of the polytope's boundary, the linear condition in Eq.~(\ref{e:linear_con}) would be still more efficient than the linear inequality constraints that is needed to realize the condition $\Phi(\sigma)\subset \A(\sigma)$.

It is important to note that Eq.~(\ref{e:linear_con}) implies several necessary conditions as follows. Every face $\tau\in\partial\K$ that is in the intersection of several $d-1$-faces $\sigma_1,..,\sigma_k\in\partial\K_{d-1}$, that is, $\tau=\sigma_1\cap...\cap\sigma_k$ must satisfy
\begin{equation}\label{e:linear_con2}
\Phi(\tau)\subset\aff(\A(\sigma_1))\cap\aff(\A(\sigma_2))\cap...\cap\aff(\A(\sigma_k))=\aff(\kappa),
\end{equation}
where $\kappa=\A(\sigma_1)\cap \A(\sigma_2)\cap ... \cap \A(\sigma_k)$, and we assume here that $\A$ is consistent in the sense that the intersection $\kappa$ is a $d-k$-face, and the normals to the $d-1$-faces $\A(\sigma_j)$ are linearly independent. Intuitively, $\Phi(\tau)$ is restricted to a lower dimensional affine space $\aff(\kappa)$. For later use, let us extend the function $\A$ to be defined over all boundary faces $\partial\K$ by defining $\A(\tau)=\kappa$.

\begin{wrapfigure}{l}{0.55\textwidth}
  \begin{center}\vspace{-0.8cm}\hspace{-0.1cm}
    \includegraphics[width=0.53\textwidth]{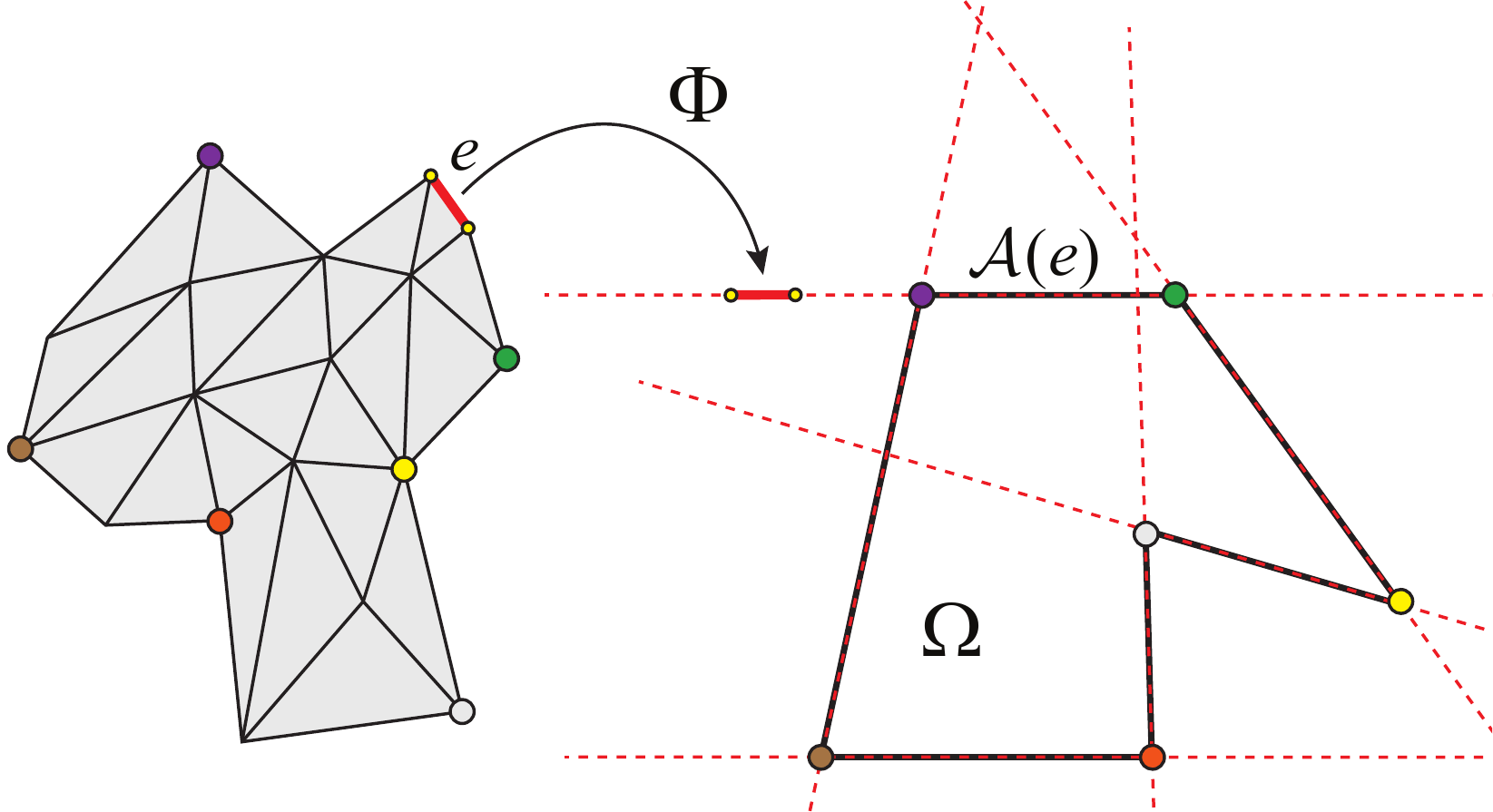}\vspace{-0.5cm}
  \end{center}
  \caption{Letting the boundary slide in the $d=2$ case.}\label{fig:slide}
\end{wrapfigure}
To give a more concrete example, imagine a triangular mesh $\M$ (i.e., $d=2$) mapped into a polygonal domain $\OOmega\subset\Real^2$ in the plane (see Figure \ref{fig:slide}). Then, Eq.~(\ref{e:linear_con}) requires each boundary edge $e\in\partial\K_1$ to be mapped somewhere on the infinite line that supports the assigned edge $\A(e)$ in the target domain's boundary, $\partial\OOmega$.
Furthermore, boundary vertices common to edges that are mapped to different (not co-linear) polygon edges are restricted to polygon's vertices. In Figure \ref{fig:slide} these are shown with colored disks.

\begin{wrapfigure}{r}{0.42\textwidth}
  \begin{center}\vspace{-1.0cm}\hspace{-0.1cm}
    \includegraphics[width=0.27\textwidth]{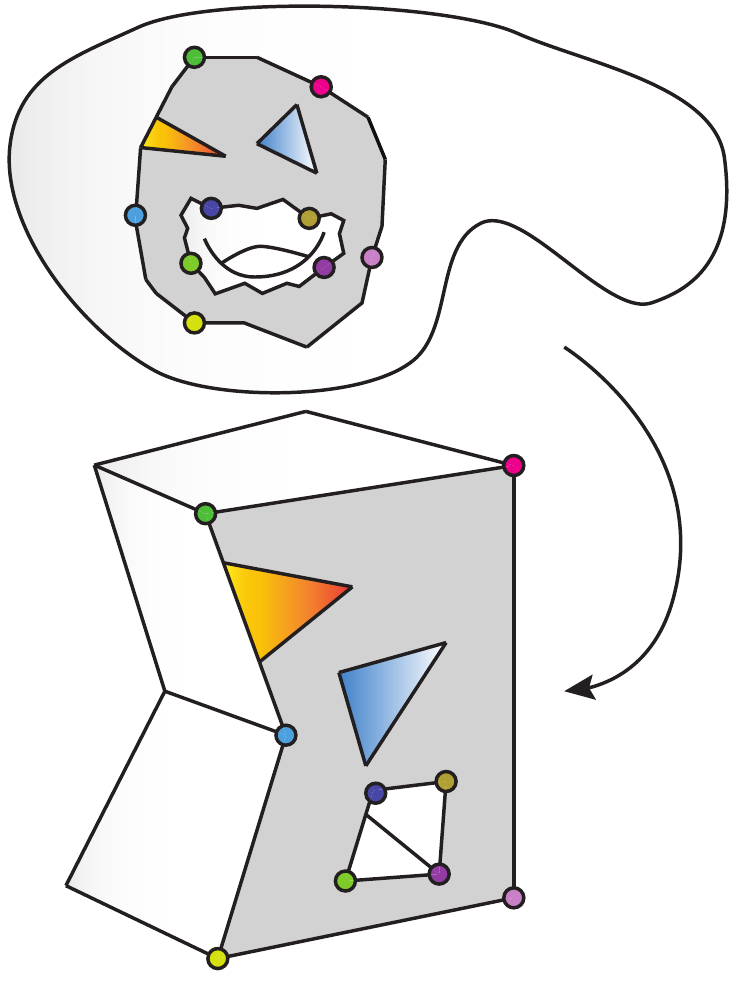}\vspace{-0.5cm}
  \end{center}
  \caption{The assignment for the $d=3$ case.}\vspace{-0.3cm}
  \label{fig:3dpolytope}
\end{wrapfigure} In the case of tetrahedral meshes (i.e., $d=3$) boundary faces $f\in \partial\K_2$ are restricted to infinite planes supporting the relevant planar polygonal faces of the polytope; boundary edges that are adjacent to not co-planar faces are mapped to inifinite lines; boundary vertices that are adjacent to three or more independent faces are mapped to fixed vertices of the polytope. Figure \ref{fig:3dpolytope} illustrates an example: it shows in grey an area of the boundary of a tetrahedral  mesh (top) that is mapped to a face of the polytope (bottom). Furthermore, it highlights the constrained vertices.

Naturally, we will need to assume that the assignment $\A$ is \emph{topologically feasible}, namely that there exists an orientation preserving boundary homeomorphism $\psi:\partial\M\too\partial\OOmega$ that satisfies the given assignment in some (arbitrary) way. Note that this homeomorphism need not be a simplicial map, and can map the boundary faces of $\partial\M$ arbitrarily onto the boundary faces of $\partial\OOmega$, all we need to make sure that topologically the provided assignment $\A$ ``makes sense''. Further note that when computing the bijective simplicial mapping in practice we do not assume we have such $\Psi$ at hand or know it in any sense, only the knowledge of its \emph{exitance} is required. In the main result of this paper we will prove that any non-degenerate orientation preserving simplicial map  $\Phi$ that satisfies Eq.~(\ref{e:linear_con}) induced by some topologically feasible assignment $\A$ is an injection over the interior of $\M$ and onto $\OOmega$:

\begin{theorem}\label{thm:bijectivity2}
Let $\M\subset \Real^n$ be a $d$-dimensional compact mesh with boundary embedded in $n\geq d$ dimensional Euclidean space, $\OOmega\subset \Real^d$ a $d$-dimensional polytope, and $\A:\partial\K_{d-1}\too\partial\OOmega_{d-1}$ a topologically feasible assignment between their boundaries. Then, any non-degenerate orientation preserving $\Phi:\M\too\Real^d$ that satisfies the linear equation (\ref{e:linear_con}) for all $\sigma\in\partial\K_{d-1}$ satisfies the following:
\begin{enumerate}
\item
$\Phi$ is injective over the interior of $\M$.
\item
$\Phi(\M)=\OOmega$.
\end{enumerate}
\end{theorem}

Note that it is a delicate point but the theorem above  does not imply that $\Phi$ is injective on the \emph{boundary} of $\M$. Although pretty rare in practice, it could happen that such a $\Phi$ is injective over $\interior{\M}$, $\Phi(\M)=\OOmega$ while $\Phi\vert_{\partial\M}$ is not injective.

\begin{wrapfigure}{r}{0.64\textwidth}
  \begin{center}\vspace{-0.4cm}\hspace{-0.4cm}
    \includegraphics[width=0.65\textwidth]{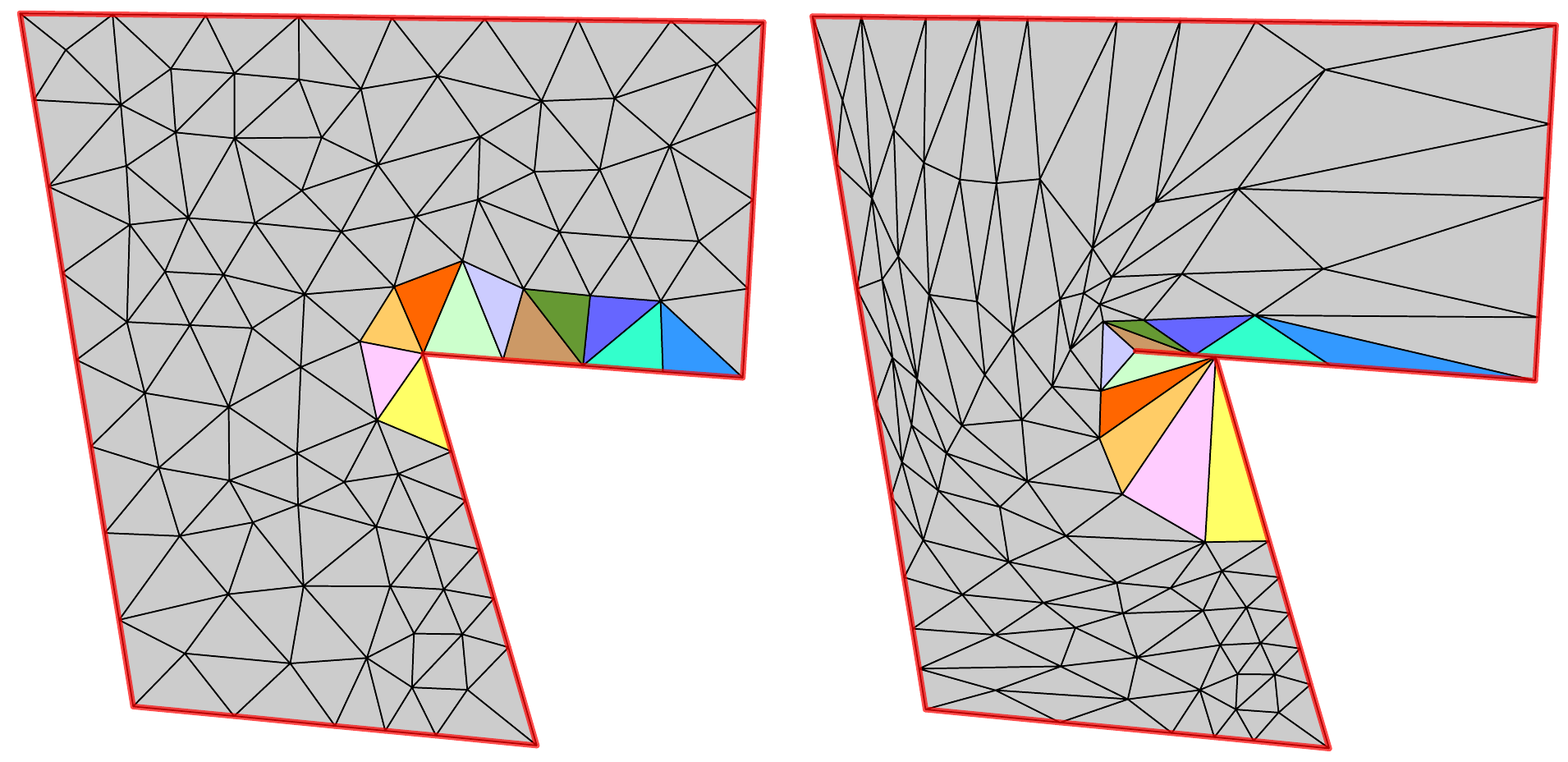}\vspace{-0.5cm}
  \end{center}
  \caption{}
  \label{fig:counter}
\end{wrapfigure}
For example, Figure \ref{fig:counter} shows a mapping of an L-shaped triangular mesh, where each boundary edge is constrained to stay within its affine hull (as Eq.~(\ref{e:linear_con}) requires), nevertheless, part of the boundary $\partial\M$ is mapped into the interior of $\OOmega$ (the entire boundary $\partial\M$ and its image $\Phi(\partial\M)$ are highlighted in red). The colors indicate corresponding triangles. In a sense what happened here is that the boundary of $\OOmega$ was extended into its interior. Note that in practice it is unlikely to come-up with such a map and here it was artificially engineered. Nevertheless, in order to guarantee that the boundary is also mapped injectively we can add the requirement that the boundary maps are orientation preserving in the following sense. First, for every boundary $d-1$-face $\tau\in\partial\OOmega_{d-1}$ all the faces $\A^{-1}(\tau)$ are mapped with their orientation preserved onto $\aff(\tau)$. Second, recursively, if $\partial\OOmega_{\ell}\ni\kappa\subset \tau\in\partial\OOmega_{\ell+1}$ then all the $\ell$-faces $\A^{-1}(\kappa)$ should preserve their orientation, where the orientation of $\kappa$ is taken to be induced by $\tau$ and the orientation of any $\sigma\in\A^{-1}(\kappa)$ is induced by a face $\alpha\in \A^{-1}(\tau)$ such that  $\sigma\subset\alpha$. A map $\Phi$ that satisfies this condition is said to be \emph{orientation preserving on the boundary}. This leads to our third and final set of sufficient conditions:

\begin{theorem}\label{thm:bijectivity3}
Let $\M\subset \Real^n$ be a $d$-dimensional compact mesh with boundary embedded in $n\geq d$ dimensional Euclidean space, $\OOmega\subset \Real^d$ a $d$-dimensional polytope, and $\A:\partial\K_{d-1}\too\partial\OOmega_{d-1}$ a topologically feasible assignment between their boundaries. Then, any non-degenerate orientation preserving $\Phi:\M\too\Real^d$ that satisfies the linear equation (\ref{e:linear_con}) for all $\sigma\in\partial\K_{d-1}$ and is orientation preserving on the boundary is a bijection between $\M$ and $\OOmega$.
\end{theorem}

In the following section we will develop the tools that are used for proving the above theorems. The main tool will be a pre-image counting argument that makes use of the power and elegance of the classical mapping degree. In short, the degree $\deg(\Phi_q, \partial\M)$ is an integer that counts how many times the simplicial map $\Phi$ wraps the boundary of the mesh $\partial\M$ around the point $q$, and this number equals almost everywhere to the number of pre-images $\#\set{\Phi^{-1}(q)}$, as the following theorem states:
\begin{theorem}\label{thm:main_inequality}
Let $\Phi:\M\too\Real^d$ be a non-degenerate orientation preserving simplicial map, and $\M$ a $d$-dimensional compact mesh with boundary. The number of pre-images $\#\set{\Phi^{-1}(q)}$ of a point $q\in\Real^d\setminus\Phi(\partial\M)$ satisfies
$$\#\set{\Phi^{-1}(q)} \leq \deg(\Phi_q, \partial\M).$$
If $q\in \Real^d \setminus \parr{\cup_{\tau\in\K_{d-1}}\Phi\parr{\tau}}$ the above inequality is replaced with an equality.
\end{theorem}

\section{The Cycle Degree and Pre-image counting argument}

Our goal in this section is to prove a useful pre-image counting argument for orientation preserving simplicial maps. This is done by employing the mapping degree tool to simplicial maps restricted to cycles. This argument will be used in the subsequent section for proving the different sufficient conditions for injectivity of simplicial maps.

We will make use of the notion of degree which in essence counts how many times a map between two closed manifold "wraps" the first manifold over the second one.

\begin{wrapfigure}{r}{0.3\textwidth}
  \begin{center}\vspace{-0.0cm}\hspace{-0.1cm}
    \includegraphics[width=0.3\textwidth]{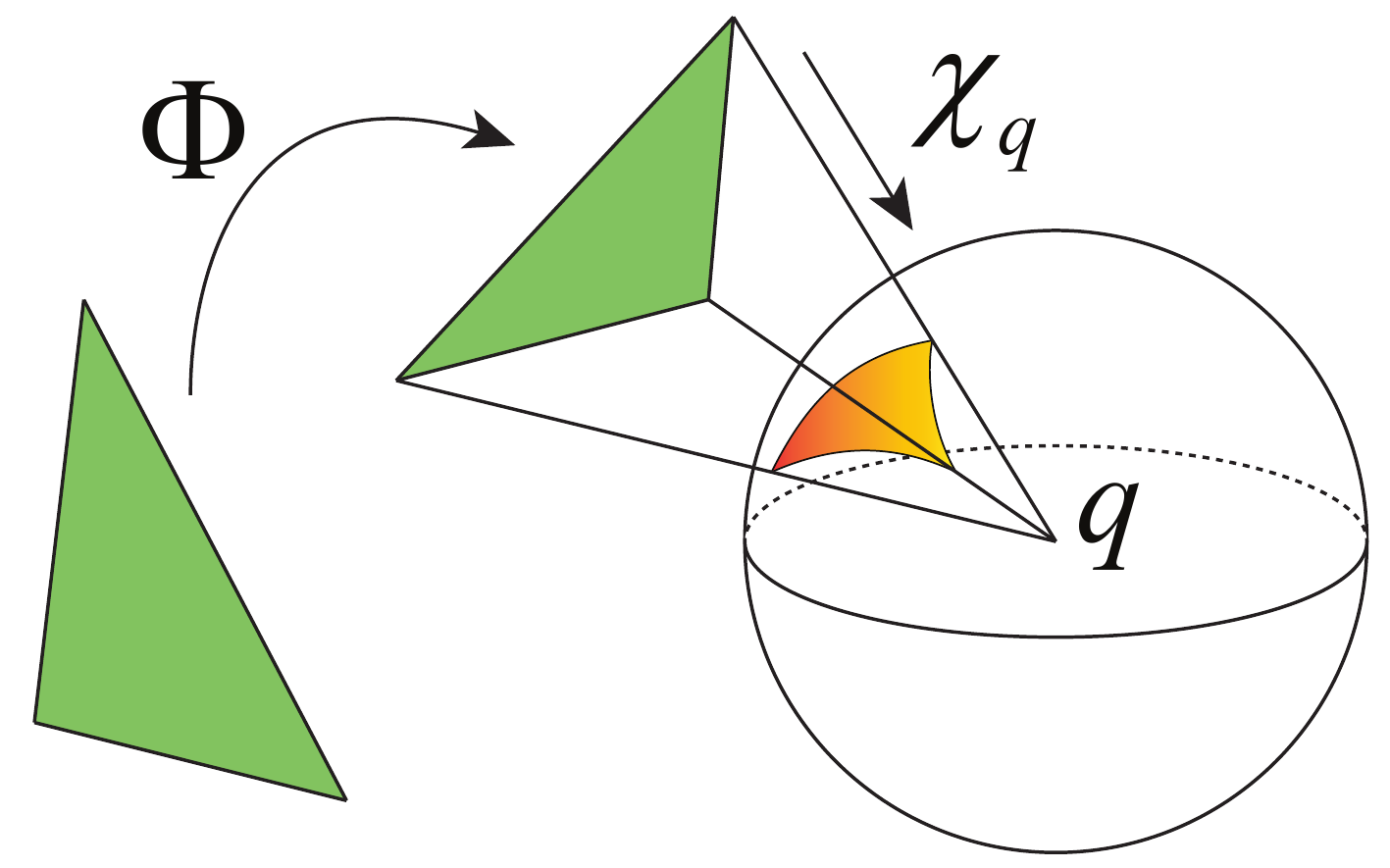}\vspace{-0.0cm}
  \end{center}
  \caption{}
  \label{fig:degree}\vspace{0.5cm}
\end{wrapfigure}

Let us define the projection onto a sphere $\chi_q:\Real^d \setminus\set{q}\too \sphere_q$, where $\sphere_q$ is the unit $d-1$-sphere centered at $q$, that is $\sphere_q=\set{p\in\Real^{d}\vert \norm{p-q}_2 = 1}$, by
$$\chi_q(p)=\frac{p-q}{\norm{p-q}_2},$$ where $\norm{\cdot}_2$ denotes the Euclidean norm in $\Real^d$. The key player in the upcoming theory is the composed map $\Phi_q=\chi_q\circ\Phi$. See Figure \ref{fig:degree} for an illustration.

\begin{wrapfigure}{l}{0.55\textwidth}
  \begin{center}\vspace{-0.5cm}\hspace{-0.3cm}
    \includegraphics[width=0.55\textwidth]{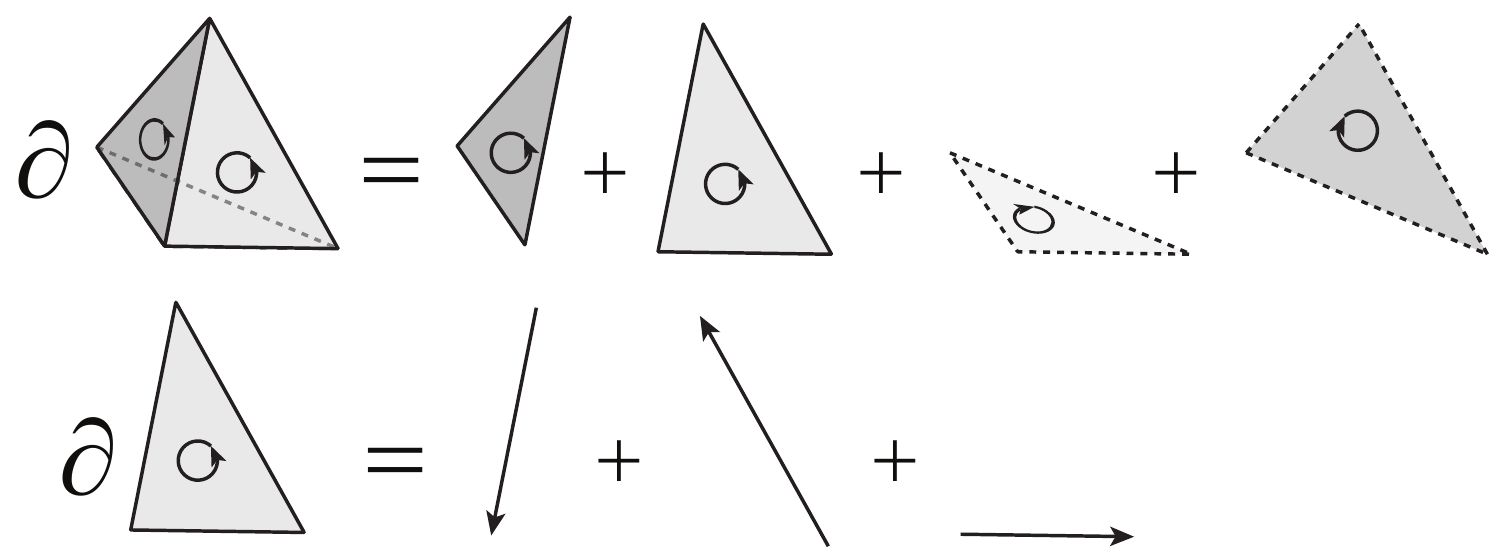}\vspace{-0.5cm}
  \end{center}
  \caption{}
  \label{fig:boundary}
\end{wrapfigure}
Before we define the notion of degree over cycles let us recall some terminology. An $\ell$-chain is a formal sum of $\ell$-faces $c=\sum_i a_i \sigma_i$, $\sigma_i\in\K_\ell$, where $a_i\in\Integer$. We denote the free abelian group of all $\ell$-chains by the symbol $\C_\ell$. Let $\partial_\ell:\C_\ell\too\C_{\ell-1}$ be the boundary operator taking $\ell$-chains to $\ell-1$-chains. Figure \ref{fig:boundary} shows an example of the boundary operator applied to $\sigma\in\K_3$ (tet, top) and to $\sigma\in\K_2$ (triangle, bottom). An $\ell$-chain $c\in\C_\ell$ is called an $\ell$-cycle if $\partial c =0$. For example, $\partial\sigma$, where $\sigma\in\K_d$, is a $d-1$-cycle as can be verified from Figure \ref{fig:boundary} for $d=3$. The sub-group of $\ell$-cycles is denoted $\ker\partial_\ell$.

We will be interested in $d-1$-cycles. Namely, $c=\sum_i a_i\sigma_i \in\ker \partial_{d-1}$. The reason is that closed (possibly self-intersecting) $d-1$ sub-surface meshes of $\M$ can be represented as $d-1$-cycles. For example, given a $d$-face $\sigma\in\K_d$, its boundary $c=\partial\sigma$ is an example of such $d-1$ sub-surface mesh and as indicated above is also a $d-1$-cycle. Furthermore, any (formal) sum of such elements, that is $c=\sum_i\partial\sigma_i$ is also an example. In, fact this is the most general type of cycles we will need to consider. We denote by $I(c)=\set{i  \vert a_i\ne 0}$ the index set of non-zero coefficients of the cycle $c=\sum_i a_i\sigma_i$.

Our goal is to define $\deg(\Phi_q,c)$ which is intuitively the degree of the map $\Phi_q$ restricted to the cycle $c$, or equivalently, how many times $\Phi_q$ wraps $c$ over $\sphere_q$. Although the notion of degree is well established for mappings between manifolds and piecewise-linear manifolds (see \cite{outerelo2009mapping} for historic overview as-well as state of the art report) we could not find in literature a direct treatment of the degree of mapping restricted to cycles. Since this notion seems very natural for analyzing simplicial mappings of meshes we develop it in detail here. We will adopt the so-called de-Rham point of view using differential forms. The reason is that it seems to give a more efficient way to define the degree on cycles and to prove its properties. It might be useful to note, however, that in the rest of this paper will only use the properties of the degree as summarized in Proposition \ref{prop:degree_properties} below.

That is, let $\omega$ be the normalized volume form on $\sphere_q$, i.e., $\int_{\sphere_q} \omega = 1$. Then we define the degree of $\Phi_q$ restricted to a $d-1$- cycle $c=\sum_i a_i \sigma_i$, $\sigma_i\in\K_{d-1}$ via
\begin{equation}\label{def:degree}
\deg(\Phi_q,c) = \int_c \Phi_q^*\,\omega = \sum_i a_i \int_{\sigma_i} \Phi^*_q \,\omega.
\end{equation}
Note that $\int_{\sigma_i} \Phi^*_q \,\omega$ is simply the (normalized) signed area of $\Phi_q(\sigma_i)$. That is, the sign of  $\int_{\sigma_i} \Phi^*_q \,\omega$ is $+1$ if $\Phi_q$ preserves the orientation of $\sigma_i$ when mapping it to the sphere, and $-1$ if it inverts it.

In order to see that Eq.~(\ref{def:degree}) actually well defines an integer and give an alternative way of computing this number we will adapt some arguments from classical degree theory of smooth mappings. For the sake of being self-contained we repeat some of the argumentation in our context. First, let $\Gamma^{d-1}(\sphere_q)$ denote the linear space of differential $d-1$-forms on the $d-1$-sphere $\sphere_q$. Given $\omega\in\Gamma^{d-1}(\sphere_q)$ we want to consider the relation between the two real scalars $\int_{\sphere_q} \omega$ and $\int_c \Phi_q^*\,\omega=\sum_i a_i \int_{\sigma_i} \Phi^*_q \,\omega$. Let us denote this relation by $T:\Real\too\Real$. Namely we define $T$ via $T(\int_{\sphere_q}\omega)=\int_c \Phi_q^*\omega$.

First, to see that this relation is well-defined take $\omega_1,\omega_2\in\Gamma^{d-1}(\sphere_q)$ such that $\int_{\sphere_q}\omega_1=\int_{\sphere_q}\omega_2$. Since $\int_{\sphere_q}\omega_1-\omega_2=0$ and $\sphere_q$ is connected, compact, and oriented manifold there exists $\eta\in\Gamma^{d-2}(\sphere_q)$ such that $\omega_1-\omega_2=d\eta$ (the $d-1$ co-homology group in this case $H^{d-1}(\sphere_q)\cong \Real$, see for example pages 268--269 and Theorem 9  in \cite{spivak1979comprehensive}).

Hence,
\begin{eqnarray*}
\int_c\Phi_q^*\omega_1 - \int_c\Phi_q^*\omega_2 &=& \sum_i a_i \int_{\sigma_i} \Phi_q^*(\omega_1-\omega_2) \\
&=& \sum_i a_i \int_{\sigma_i} \Phi_q^*d\eta \\
&\overset{\text{$d$ commutes with pull-back}}=& \sum_i a_i \int_{\sigma_i} d\Phi_q^*\eta \\
&\overset{\text{Stokes}}=& \sum_i a_i \int_{\partial\sigma_i} \Phi_q^*\eta \\
&\overset{\text{$c$ is a cycle}}=& 0,
\end{eqnarray*}
where we used the fact that the operator $d$ commutes with the pull-back operation, Stokes theorem, and the fact that $c$ is a cycle. The fact that $H^{d-1}(\sphere_q)\cong \Real$ also implies that for every $\alpha\in\Real$ there exists $\omega\in\Gamma^{d-1}(\sphere_q)$ such that $\int_{\sphere_q}\omega=\alpha$. That is, the domain of $T$ is the whole real line $\Real$, and it is well-defined.

We have seen that $T:\Real\too\Real$ is well-defined over the whole real line. Now we show it is linear. Taking $\omega_1,\omega_2\in\Gamma^{d-1}(\sphere_q)$, and scalars $\alpha,\beta\in\Real$ we have for  $\omega=\alpha\omega_1+\beta\omega_2$ that
\begin{eqnarray*}
T\parr{\alpha\int_{\sphere_q}\omega_1+\beta\int_{\sphere_q}\omega_2}&=& T\parr{\int_{\sphere_q} \alpha\omega_1+\beta\omega_2} \\ &=&\int_c\Phi_q^*\parr{\alpha\omega_1+\beta\omega_2}\\
\\ &=&\alpha\int_c\Phi_q^*\omega_1 + \beta\int_c\Phi_q^*\omega_2\\
&=& \alpha T\parr{\int_{\sphere_q}\omega_1} + \beta T\parr{\int_{\sphere_q}\omega_2}.
\end{eqnarray*}
This implies that $T:\Real\too\Real$ is linear so it has to be of the form
\begin{equation}\label{eq:T_linear_d}
\int_c\Phi_q^*\omega = T\parr{\int_{\sphere_q}\omega} = d \int_{\sphere_q}\omega,
\end{equation}
where $d$ is some constant independent of the choice of $w\in\Gamma^d(\sphere_q)$. Since this is true for all forms in $\omega\in\Gamma^d(\sphere_q)$ we get that $d=\deg(\Phi_q,c)$ as defined in Eq.~(\ref{def:degree}). It might seem that we have not accomplished anything with this definition of $T$, however, having at our disposal Eq.~(\ref{eq:T_linear_d}) (with constant $d=\deg(\Phi_q,c)$, regardless of the choice of $\omega$) will help us easily prove all the nice property of the degree on cycles. In particular, we next prove that the degree (although still not clear from the above definition) is always an integer, coincides with the classical Brouwer degree in case the cycle $c$ represent a polyhedral surface,  and show another useful way to calculate it.

\begin{proposition}\label{prop:degree_properties}
Let $\Phi:\M\too\Real^d$ be a simplicial map of a $d$-dimensional compact mesh $\M$, $\chi_q:\Real^d\setminus\set{q}\too \sphere_q$ a projection on the sphere centered at $q$, and $\Phi_q=\chi_q\circ\Phi$ their composition. Furthermore, let $c=\sum_i a_i\sigma_i\in\ker\partial_{d-1}$ be a $d-1$-cycle in $\M$. The number $\deg(\Phi_q,c)$ defined in Eq.~(\ref{def:degree}) satisfies the following properties:

\begin{enumerate}

\item \label{item:deg_integer}
It is an integer.

\item \label{item:brouwer}
In case the cycle $c$ represents a $d-1$-piecewise-linear manifold (e.g., a polygon when $d=2$, and a polyhedral surface when $d=3$) this notion of degree coincides with the classical Brouwer's degree.

\item \label{item:alternative_formula_deg}
For any $p\in\sphere_q\setminus \bigcup_{i\in I(c)}\Phi_q(\partial\sigma_i)$, $$\deg(\Phi_q,c)=\sum_{i\in I(c) \ s.t. \ p\in \Phi_q(\sigma_i) }a_i \, \sign_{\sigma_i}(\Phi_q),$$ where $\sign_{\sigma_i}(\Phi_q)=1$ if $\Phi_q$ preserves the orientation of $\sigma_i$ as it maps it onto the sphere $\sphere_q$, and $\sign_{\sigma_i}(\Phi_q)=-1$ if it inverts it, and $I(c)=\set{i  \vert a_i\ne 0}$.

\item\label{item:additivity_of_deg}
Let $c'=\sum_i a'_i \sigma_i \in \ker\partial_{d-1}$ be another $d-1$-cycle in $\M$ then $$\deg(\Phi_q,c+c')=\deg(\Phi_q,c)+\deg(\Phi_q,c').$$

\end{enumerate}
\end{proposition}

\begin{proof}
We start with proving Property (\ref{item:alternative_formula_deg}).\\
Take some $p\in\sphere_q\setminus \bigcup_{i\in I(c)}\Phi_q(\partial\sigma_i)$. Denote its ``pre-images'' on $c$ under $\Phi_q$ by $x_1,x_2,..,x_n$. That is, there exists some $\sigma_{i_j}$, $i_j\in I(c)$, $j=1..n$, such that $p=\Phi_q(x_j)$, and $x_j\in \sigma_{i_j}$.

If there are no such pre-images to $p$ then we can take a form $\omega\in\Gamma^{d-1}(\sphere_q)$ such that $\int_{\sphere_q}\omega>0$ and $\omega=0$ outside a small neighborhood $V$ of $p$, where all points in $V$ do not have pre-images on $c$ under $\Phi_q$ (existence of such $V$ is implied by the continuity of $\Phi_q$ and the compactness of $\M$). Using $\omega$ in Eq.~(\ref{eq:T_linear_d}) shows that $d=\deg(\Phi_q,c)=0$. This implies Property (\ref{item:alternative_formula_deg}) in this case.

We now assume there exists at-least one pre-image, i.e.  $n\geq 1$. There exist some neighborhood $V$ of $p$, $p\in V\subset \sphere_q$, and disjoint neighborhoods $U_j$ of $x_j$, $x_j\in U_j \subset \sigma_{i_j}$ such that $\Phi_q^{-1}(V)=\cup_{j=1}^n U_j$, and $\Phi_q(U_j)=V$. Take some $\omega\in\Gamma^{d-1}(\sphere_q)$ such that $\int_{\sphere_q}\omega>0$ and $\omega=0$ outside  $V$. Then,
\begin{eqnarray*}
\int_{c} \Phi_q^*\omega &=& \sum_i a_{i} \int_{\sigma_i} \Phi_q^*\omega \\ &=& \sum_{j=1}^n a_{i_j} \int_{U_j} \Phi_q^*\omega \\ &=& \sum_{j=1}^n a_{i_j} \sign_{\sigma_{i_j}}(\Phi_q)\parr{\int_V \omega}.
\end{eqnarray*}
On the other hand Eq.~(\ref{eq:T_linear_d}) implies that $\int_{c} \Phi_q^*\omega=\parr{\int_V \omega}  \deg(\Phi_q,c)$. Combining and dividing by $\int_V \omega$ we get property (\ref{item:alternative_formula_deg}).

Property (\ref{item:deg_integer}) follows from Property (\ref{item:alternative_formula_deg}) since all $a_i$ and $\sign_\sigma(\Phi_q)$ are integers so the degree is an integer.

To prove property (\ref{item:additivity_of_deg}) we will simply use the definition in Eq.~(\ref{def:degree}): let $c'=\sum_i a_i'\sigma_i$ be a $d-1$-cycle, and $\omega$  the normalized volume form on $\sphere_q$,
\begin{eqnarray*}
\deg(\Phi_q,c+c') &=& \sum_i \parr{a_i+a'_i} \int_{\sigma_i} \Phi^*_q \,\omega \\  &=& \sum_i a_i \int_{\sigma_i} \Phi^*_q \,\omega + \sum_i a'_i \int_{\sigma_i} \Phi^*_q \,\omega \\ &=& \deg(\Phi_q,c)+\deg(\Phi_q,c').
\end{eqnarray*}
Lastly, property (\ref{item:brouwer}) can be proved by noting that property (\ref{item:alternative_formula_deg}) coincides with the classical definition of Brouwer's degree in the case $c$ is a cycle representing a $d-1$ dimensional piecewise-linear manifold.
\end{proof}

In the next lemma we compute the cycle degree for a simple $d-1$-cycle, that is the boundary of a $d$-face.

\begin{lemma}\label{lem:degree_d_face}
Let $\Phi:\M\too\Real^d$ be a non-degenerate orientation preserving simplicial map of a $d$-dimensional compact mesh $\M$.  For all $\sigma\in\K_d$ such that $q\notin\Phi(\partial\sigma)$ there exists $$\deg(\Phi_q,\partial\sigma) = \begin{cases} 1 & q\in\Phi(\sigma) \\ 0 & q\notin\Phi(\sigma) \end{cases}.$$
\end{lemma}
\begin{proof}
If $q\notin\Phi(\sigma)$ there exist a hyperplane separating the point $q$ and the convex set (tet) $\Phi(\sigma)$. This implies that at-least half of the sphere $\sphere_q$ has no pre-images by $\Phi_q$. Proposition \ref{prop:degree_properties}, property (\ref{item:alternative_formula_deg}) then implies that $\deg(\Phi_q,\partial\sigma)=0$.\\ In the case $q\in\interior{\Phi(\sigma)}$, one can pick a point $p\in\sphere_q$ such that $p\in\interior{\Phi_q(\tau)}$, where $\tau\subset\partial\sigma$. Since $\Phi(\sigma)$ is convex, $\tau$ is unique. Since $\Phi\vert_\sigma$ is orientation preserving and non-degenerate $\sign_\tau(\Phi_q)=1$. Using Proposition \ref{prop:degree_properties}, property (\ref{item:alternative_formula_deg}) again implies  $\deg(\Phi_q,\partial\sigma)=1$.

\end{proof}

Before we prove the main result of this section, namely the pre-image counting argument, let us mention a useful property of orientation preserving simplicial maps, namely that they are \emph{open} maps. That is, mapping open sets to open sets.
This property, that is proved in the appendix is used later on to extend injective mappings over zero-measure sets.

\begin{lemma}\label{lem:open_map}
Let $\Phi:\M\subset \Real^n\too \Real^d$, $n\geq d$, be a non-degenerate orientation preserving simplicial map of a compact $d$-dimensional mesh $\M$ into $\Real^d$. Then $\Phi$ is an open map.
\end{lemma}

%
%
%

We now get to the main theorem of this section:

\textbf{Theorem \ref{thm:main_inequality}.}\textit{
Let $\Phi:\M\too\Real^d$ be a non-degenerate orientation preserving simplicial map, and $\M$ a $d$-dimensional compact mesh with boundary. The number of pre-images $\#\set{\Phi^{-1}(q)}$ of a point $q\in\Real^d\setminus\Phi(\partial\M)$ satisfies
$$\#\set{\Phi^{-1}(q)} \leq \deg(\Phi_q, \partial\M).$$
If $q\in \Real^d \setminus \parr{\cup_{\tau\in\K_{d-1}}\Phi\parr{\tau}}$ the above inequality is replaced with an equality.}

Before proving the theorem, note that the inequality can be strict. For example, consider the example in Figure \ref{fig:counter_example}: if we take $q$ to be the central vertex on the right image in this example, then it has only one pre-image but its degree is $\deg(\Phi_q,\partial\M)=2$ as this is precisely the winding number of the image boundary polygon w.r.t. the central vertex $q$.

\begin{proof}
Let us start with the second part of the theorem.
Denote the set $Y=\cup_{\tau\in\K_{d-1}}\Phi\parr{\tau}$. Let $q\in\Real^d\setminus Y$. Lemma \ref{lem:degree_d_face} and Proposition \ref{prop:degree_properties}- Property (\ref{item:additivity_of_deg}) imply:
\begin{eqnarray*}\label{eq:preimage_eq1}
\#\set{\Phi^{-1}(q)} &=& \sum_{\sigma\in\K_d \, \mathrm{s.t} \, q\in\Phi(\sigma) }1 \\&\overset{\text{Lemma} \ref{lem:degree_d_face}}=& \sum_{\sigma \in \K_d }\deg(\Phi_q,\partial \sigma)\\ &\overset{\text{Prop. \ref{prop:degree_properties}, Property (\ref{item:additivity_of_deg})}}=&\deg(\Phi_q,\partial\M).
\end{eqnarray*}
This proves the second claim of the theorem. Now we use Lemma \ref{lem:open_map} to prove the first claim. Assume the claim is not true, that is there exists $q\in \Real^d\setminus\Phi(\partial\M)$ such that $\#\set{\Phi^{-1}(q)} > \deg(\Phi_q, \partial\M)$. By Lemma \ref{lem:open_map} the map $\Phi$ is open so for every $x\in\set{\Phi^{-1}(q)}$ we can take an open neighborhood $U_x$ (where all such neighborhoods $U_x$ are pairwise disjoint), and $V=\cap_{x\in \set{\Phi^{-1}(q)}}\Phi(U_x)$ (note that there is a finite number of pre-images so this is a finite intersection) is an open neighborhood of $q$. Every point $q'\in V$ has strictly more pre-images than $\deg(\Phi_q, \partial\M)$. Since $q\cap \Phi(\partial\M)=\emptyset$ we can take $V$ sufficiently small so that $V\cap \Phi(\partial\M)=\emptyset$. Since $\deg(\Phi_{q'}, \partial\M)$ is an integer and continuous when $q'\in V$ it has to be constant in $V$. Since $Y$ is of measure zero we can find a point $q'\in V \setminus Y$. This point has more pre-images than the integer $\deg(\Phi_q,\partial\M)=\deg(\Phi_{q'},\partial\M)$ in contradiction to the claim already proven above. This concludes the proof.\end{proof}


\section{Sufficient conditions for Injectivity}

Let us quickly recall our setting. We are interested to map a $d$-dimensional mesh embedded in $n\geq d$ dimensional Euclidean space, $\M\subset \Real^n$, injectively and  onto a target domain in the form of a $d$-dimensional polytope $\OOmega\subset \Real^d$. In this section we prove three sets of sufficient conditions. Of-course, all the sets of sufficient conditions ask the candidate simplicial map $\Phi:\M\too\Real^d$ to satisfy the necessary condition for injectivity, namely that it is non-degenerate and has consistent orientation (see Proposition \ref{prop:necessary}). As stated above we will restrict our attention to orientation preserving simplicial maps while keeping in mind that the orientation reversing case is similar.

The first set of sufficient conditions generalizes the global inversion theorem from classical analysis to the piecewise-linear mesh case and simply states that in addition to the necessary non-degeneracy and orientation preserving condition, the boundary map should be bijective:

\textbf{Theorem \ref{thm:bijectivity1}.} \textit{A non-degenerate orientation preserving simplicial map $\Phi:\M\too\Real^d$ of a $d$-dimensional compact mesh with boundary $\M$ is a bijection $\Phi:\M\too\OOmega$ if the boundary map $\Phi\vert_{\partial\M}:\partial\M\too \partial\OOmega$ is bijective.}
\begin{proof}
For an arbitrary point $q\in\interior{\OOmega}$, by assumption $q\notin\Phi(\partial\M)$, and Theorem \ref{thm:main_inequality} implies that $\#\set{\Phi^{-1}(q)}\leq\deg(\Phi_q,\partial\M)=1$. On the other hand for $q$ in the dense set $\interior{\OOmega}\setminus\cup_{\tau\in\K_{d-1}}\Phi(\tau)$
we have $\#\set{\Phi^{-1}(q)}=1$. By continuity of $\Phi$ and compactness of $\M$ this implies that $\#\set{\Phi^{-1}(q)}\geq 1$ for all $q\in\interior{\OOmega}$. Therefore, $\#\set{\Phi^{-1}(q)}= 1$ for all $q\in\interior{\OOmega}$.
For arbitrary $q\in\Real^d\setminus\OOmega$, Theorem \ref{thm:main_inequality} implies that $\#\set{\Phi^{-1}(q)}\leq\deg(\Phi_q,\partial\M)=0$.
This concludes the proof.
\end{proof}

We now move to prove the second set of sufficient conditions, namely Theorem \ref{thm:bijectivity2}. In addition to the setting introduced in the beginning of this section we are also given a topologically feasible assignment $\A:\partial\K_{d-1}\too\partial\OOmega_{d-1}$ of boundary faces of $\M$ to boundary faces of the target polytope domain $\OOmega$. We consider any non-degenerate orientation preserving $\Phi:\M\too\Real^d$ that satisfies the linear equation (\ref{e:linear_con}) for all $\sigma\in\partial\K_{d-1}$ and we will prove that such a $\Phi$ is injective over $\interior{\M}$ and covers $\OOmega$, that is $\Phi(\M)=\OOmega$.
The main power in this formulation of sufficient conditions is that it allows considering a \emph{collection} of different boundary mappings at the minor price of adding linear constraints to the problem.
The idea of the proof is to reduce to the case of Theorem \ref{thm:bijectivity1} using a homotopy argument. We start with calculating the degree $\deg(\Phi_q,\partial\M)$ for \emph{almost all} $q\in\Real^d$.
\begin{lemma}\label{lem:deg_with_sliding_bndry}
Let $\M\subset \Real^n$ be a $d$-dimensional mesh with boundary embedded in $n\geq d$ dimensional Euclidean space, $\OOmega\subset \Real^d$ a $d$-dimensional polytope, and $\A:\partial\K_{d-1}\too\partial\OOmega_{d-1}$ a topologically feasible assignment. Then, for any non-degenerate orientation preserving $\Phi:\M\too\Real^d$ that satisfies the linear equation (\ref{e:linear_con}) and for all $\sigma\in\partial\K_{d-1}$ there exists $$\deg(\Phi_q, \partial\M) =  \begin{cases}1 &  q\in \OOmega\setminus Z\\ 0 & q\in\Real^d\setminus\parr{\OOmega\cup Z} \end{cases},$$ where $Z=\bigcup_{\omega\in \partial\OOmega_{d-1}}\aff(\omega)$ is the union of all affine hyperplanes supporting the boundary $d-1$-faces of $\OOmega$.
\end{lemma}
A visualization of the set $Z$ can be seen in Figure \ref{fig:slide} as the union of all red dashed lines.
\begin{proof}
We will prove the lemma by employing a degree theory argument. By Proposition \ref{prop:degree_properties} our definition of the number $\deg(\Phi_q, \partial\M)$ coincides with the well-known Brouwer's degree of the continuous map $\Phi_q\vert_{\partial \M}:\partial\M\too\sphere_q$. Since we assume $\A$ is topologically feasible there exists an orientation preserving homeomorphism (not necessarily a simplicial map) $\Psi:\partial\M\too\partial\OOmega$ such that $\Psi(\sigma)\subset \A(\sigma)$ for all
$\sigma\in\partial\K_{d-1}$. Let us consider the family of mappings $\varphi^t_q(\cdot)$ defined by
$$\varphi^t_{q}(x)=\chi_q\Big(\parr{1-t}\Phi(x)+t\Psi(x)\Big), \ \ t\in[0,1], \ x\in \partial\M.$$
Let $Z=\bigcup_{\omega\in \partial\OOmega_{d-1}}\aff(\omega)$. We claim that, for all $q\in\Real^d\setminus Z$, $\varphi^t_q:\partial\M\too\sphere_q$ is a homotopy.  Indeed, fix any $q\in\Real^d\setminus Z$ and consider an arbitrary $x\in\partial\M$. Since $x\in\sigma$ for some $\sigma\in\partial\K_{d-1}$ it means that $\Phi(x)\in\Phi(\sigma)\subset\aff(\A(\sigma))$ and $\Psi(x)\in\Psi(\sigma)\subset \A(\sigma)\subset\aff(\A(\sigma))$. This implies that $$\parr{1-t}\Phi(x)+t\Psi(x)\in \aff(\A(\sigma))\subset Z.$$ Since $Z$ is closed, $q$ has some positive distance to $Z$ and we conclude that $\varphi^t_q$ is continuous in $t$ and $x$, and hence a homotopy. The invariance of the degree to homotopy now implies that for all $q\in\Real^d\setminus Z$, $$\deg(\Phi_q,\partial\M)=\deg(\chi_q\circ\Psi,\partial\M),$$ and since $\Psi:\partial\M\too\partial\OOmega$ is an orientation preserving homeomorphism, for any interior point $q\in\OOmega\setminus Z$,  $\deg(\chi_q\circ\Psi,\partial\M)=1$, and for any $\Real^d\setminus\parr{\OOmega\cup Z}$ it equals $0$, as required.
\end{proof}

We now prove the second set of sufficient conditions.

\textbf{Theorem \ref{thm:bijectivity2}.} \textit{Let $\M\subset \Real^n$ be a $d$-dimensional compact mesh with boundary embedded in $n\geq d$ dimensional Euclidean space, $\OOmega\subset \Real^d$ a $d$-dimensional polytope, and $\A:\partial\K_{d-1}\too\partial\OOmega_{d-1}$ a topologically feasible assignment between their boundaries. Then, any non-degenerate orientation preserving $\Phi:\M\too\Real^d$ that satisfies the linear equation (\ref{e:linear_con}) for all $\sigma\in\partial\K_{d-1}$ satisfies the following:
\begin{enumerate}
\item
$\Phi$ is injective over the interior of $\M$.
\item
$\Phi(\M)=\OOmega$.
\end{enumerate}}
\begin{proof}
We start with proving (2).\\
Let $Y=\cup_{\tau\in\K_{d-1}}\Phi\parr{\tau}$. Theorem \ref{thm:main_inequality} implies that $\#\set{\Phi^{-1}(q)}=\deg(\Phi_q,\partial\M)$ for $q\in\Real^d\setminus Y$. Let $Z=\bigcup_{\omega\in \partial\OOmega_{d-1}}\aff(\omega)$. Lemma \ref{lem:deg_with_sliding_bndry} now implies that
\begin{equation}\label{thm_eq:preimages_count}
\#\set{\Phi^{-1}(q)} =  \begin{cases}1 &  q\in \OOmega\setminus \parr{Z\cup Y}\\ 0 & q\in\Real^d\setminus\parr{\OOmega\cup Z\cup Y} \end{cases}.
\end{equation}
Denote $Q=Z\cup Y$. Note that $Q$ is a set of measure zero. Take an arbitrary point $q\in \Real^d\setminus \OOmega$. If $q$ is outside $Q$ we already showed it has no pre-images in $\M$. Assume that $q\in Q$, and assume it has a pre-image $x\in\M$. $x\in\sigma$ for some $\sigma\in\K_d$, and since $\Phi$ is non-degenerate, the set $\Phi(\sigma)$ is $d$-dimensional face and therefore we can find $w\in \Phi(\sigma)\setminus\parr{\OOmega\cup Q}$ with a pre-image in $\sigma$, in contradiction with Eq.~(\ref{thm_eq:preimages_count}). Therefore $\#\set{\Phi^{-1}(q)}=0$ for any $q\in\Real^d\setminus\OOmega$. This implies that $\Phi(\M)\subset\OOmega$. Let us show that $\OOmega\subset\Phi(\M)$. First from Eq.~(\ref{thm_eq:preimages_count}) we have $\OOmega\setminus Q \subset \Phi(\M)$. For any point $q\in \OOmega \cap Q$ we can take a series $q_n\too q$ such that $\#\set{\Phi^{-1}(q_n)}=1$ (using Eq.~(\ref{thm_eq:preimages_count})). Let $\set{x_n}\subset\M$ be their pre-images, that is $\Phi(x_n)=q_n$. Since $\M$ is compact we can extract a convergent subsequence $x_n\too x\in\M$ (we abused notation and kept the original indexing). Now from continuity of $\Phi$ we have $\Phi(x)=\Phi(\lim(x_n))=\lim\Phi(x_n)=\lim q_n=q$, and we have found a pre-image of $q$ and therefore $q\in\Phi(\M)$. Since this is true for all $q\in\OOmega\cap Q$, we have $\OOmega\subset\Phi(\M)$ and therefore $\Phi(\M)=\OOmega$.

Let us prove the injectivity property (1) next. Assume in negation that there exist two points $x,x'\in\interior{\M}$ such that $\Phi(x)=\Phi(x')$. Denote $q=\Phi(x)$. By the open map property of $\Phi$ (see Lemma \ref{lem:open_map}) there exists a neighborhood $U$ of $q$ such that all $q'\in U$ have at-least two pre-images in $\M$. This contradicts Eq.~(\ref{thm_eq:preimages_count}) and concludes the proof. \end{proof}

As a side remark note that the open map property and Eq.~(\ref{thm_eq:preimages_count}) above also imply that $\Phi(\interior{\M})\subset \interior{\OOmega}$.

As mentioned and demonstrated in Section \ref{s:preliminaries}, Theorem \ref{thm:bijectivity2} guarantees the injectivity of the map $\Phi$ in the interior of the mesh, and that $\Phi$ is onto the domain $\OOmega$ but does not guarantee that the boundary is mapped bijectively. In the next and final set of sufficient conditions we show that adding the condition that $\Phi$ is orientation preserving on the boundary (as defined in Section \ref{s:preliminaries}) leads to injectivity of $\Phi$ over the boundary of $\M$ as-well.
\newpage
\textbf{Theorem \ref{thm:bijectivity3}.} \textit{Let $\M\subset \Real^n$ be a $d$-dimensional compact mesh with boundary embedded in $n\geq d$ dimensional Euclidean space, $\OOmega\subset \Real^d$ a $d$-dimensional polytope, and $\A:\partial\K_{d-1}\too\partial\OOmega_{d-1}$ a topologically feasible assignment between their boundaries. Then, any non-degenerate orientation preserving $\Phi:\M\too\Real^d$ that satisfies the linear equation (\ref{e:linear_con}) for all $\sigma\in\partial\K_{d-1}$ and is orientation preserving on the boundary is a bijection between $\M$ and $\OOmega$.}
\begin{proof}
The proof is by induction on $d$.\\
Let us prove the theorem for $d=1$. That is, $\M$ is a polygonal line defined by a finite series of ordered points $x_0,x_1,x_2,...,x_L \in \Real^n$, its boundary consists of the two end points $\partial\M=\set{x_0,x_L}$, and $\OOmega$ is $1$-dimensional polytope in $\Real$ homeomorphic to $\M$, namely a segment $[a,b]\subset\Real$. $\Phi$ is non-degenerate and orientation preserving which means that $\Phi(x_{i+1})>\Phi(x_i)$ for all $i=0,...,L-1$. Furthermore $\Phi$ satisfies Eq.~(\ref{e:linear_con}) which in this case means that $\Phi(x_0)=a$, and $\Phi(x_L)=b$. Putting this together we have $a=\Phi(x_0)<\Phi(x_1)<...<\Phi(x_{L-1})<\Phi(x_L)=b$ that implies that $\Phi$ is a bijection onto $[a,b]$.

Let us assume we have proved the $d=\ell-1$ case, and prove the $d=\ell$ case. We have an $\ell$-dimensional mesh $\M$, and a polytope $\OOmega$ with boundary $\partial\OOmega$ of dimension $\ell-1$. Consider an arbitrary polytope's boundary $\ell-1$-face $\tau\in\partial\OOmega_{\ell-1}$. Let $\M_\tau\subset\partial\M$ be the $\ell-1$ dimensional submesh that is assigned to be mapped into $\aff(\tau)$. Remember that by our definition of the assignment function $\A$ we have that the set of $\ell-1$-faces of $\M_\tau\subset\partial\M$ that is mapped into $\aff(\tau)$ is $\A^{-1}(\tau)$.

We know by assumption that $\Phi$ restricted to $\M_\tau$, that is $\Phi\vert_{\M_{\tau}}$, does not degenerate and preserves the orientation of all faces in $\A^{-1}(\tau)\subset\partial\K_{\ell-1}$ as it maps them into $\aff(\tau)$. $\A$ restricted to $\partial\M_\tau$ is topologically feasible (since $\A$ is topologically feasible). Furthermore, by definition $\Phi\vert_{\M_{\tau}}$ is  orientation preserving on the boundary $\partial\M_\tau$, and by Eq.~(\ref{e:linear_con2}) it satisfies Eq.~(\ref{e:linear_con}) w.r.t. the submesh $\M_\tau$ and the restricted $\A$. The induction assumption now implies that $\Phi\vert_{\M_\tau}$ is injective and onto $\tau$. As $\tau\in\partial\OOmega_{\ell-1}$ was arbitrary we have that $\Phi\vert_{\partial\M}$ is bijective, and since $\Phi$ is non-degenerate and orientation preserving, Theorem \ref{thm:bijectivity1} implies that $\Phi$ is a bijection.
\end{proof}

\vspace{-0.5cm}
\section{Numerical Experiments}
Theorems \ref{thm:bijectivity1},\ref{thm:bijectivity2},\ref{thm:bijectivity3} can be used to design algorithms that produce injective mappings of meshes onto polytopes. Of special interest is Theorem \ref{thm:bijectivity2} as it adds simple linear conditions (i.e., Eq.~(\ref{e:linear_con})) in addition to the necessary orientation preservation constraints and allows working with weaker boundary conditions than prescribing a bijective boundary map, as required from Theorem \ref{thm:bijectivity1}. Namely, it only requires providing a topologically feasible assignment $\A$. In this section we demonstrate how this can be used for mapping triangular meshes, i.e., $d=2$.

Let $\M=(\K_0\cup\K_1\cup\K_2)$ be a triangular mesh with boundary $\partial\M$, where $\K_0=\set{v_i}$ the set of vertices in $\Real^2$ or $\Real^3$, $\K_1=\set{e_k}$ the set of edges, and $\K_2=\set{f_j}$ the set of triangles. Assume we want to map $\M$ bijectively onto a planar polygonal domain $\Omega\subset \Real^2$, that is compute a bijective simplicial map $\Phi:\M\too\OOmega$. A simplicial map of $\M$ is uniquely  described by setting the image of the vertices $\Phi(v_i)$, $v_i\in\K_0$. Then, the affine map of each triangle $f_j\in\K_2$ is the unique affine map that takes the corners $v_1,v_2,v_3$ of triangle $f_j$ to $\Phi(v_1),\Phi(v_2),\Phi(v_3)$. Fixing a coordinate frame in triangle $f_j$ and $\Real^2$, we denote by $\v_i\in\Real^2$, $i=1,2,3$ the coordinate vector representing vertex $v_i$ in the frame of triangle $f_j$, and $\u_i\in\Real^2$ the unknown vector representing $\Phi(v_i)$ in the global frame in the plane. Then, the affine map of face $f_j$, written in coordinates as $\A_j(\vecx)=A_j\vecx+\delta_j$
can be expressed as the unique solution of the linear system:
$$\brac{A_j\,\, \delta_j}\left [ \begin{array}{ccc}\v_1 & \v_2 & \v_3 \\
1 & 1 & 1 \\
\end{array}\right ]= \left [\begin{array}{ccc}\u_1 & \u_2 & \u_3 \\
\end{array}\right ].$$
That is, $A_j,\delta_j$ can be written as a constant linear combination of the variables $\u_i$.

For assuring the necessary conditions for bijectivity, namely that each triangle's affine map $A_j$ is orientation preservation, one can use the convex constraints introduced, for example, in \cite{Lipman:2012:BDM,Lipman:2013:feature_matching} or \cite{Bommes:2013}. We will use the former mainly since they posses a maximality property. For completeness let us recap and reformulate the conditions here. For maintaining the orientation we ask that the determinant of the Jacobian is positive, namely  $\det(A_j)>0$ for all $f_j\in\K_2$. By direct computation one can check that $$2\det(A_j)=\norm{B_j}^2_F-\norm{C_j}^2_F,$$ where $B_j=\frac{A_j-A_j^T+\tr(A_j)I}{2}$ and $C_j = \frac{A_j+A_j^T-\tr(A_j)I}{2}$. Requiring that $\det(A_j)>0$ is equivalent to requiring that $$\norm{B_j}_F>\norm{C_j}_F.$$ As this is not a convex space we can, similarly to \cite{Lipman:2012:BDM,Lipman:2013:feature_matching} carve maximal convex subsets of this set via $$\frac{\tr(R_j^T B_j)}{\sqrt{2}}>\norm{C_j}_F,$$ where $R_j$ is an arbitrary rotation matrix (we'll explain shortly how to choose it). This is a convex second-order cone constraint which can be optimized using standard SOCP solvers such as MOSEK \cite{MOSEK}. In fact, along with the positive determinant conditions we found it to be numerically stable to bound also the condition number of $A_j$, where the condition number is defined as the ratio of the maximal to minimal singular values of $A_j$. Both can be enforced by a slight modification to the above equation:
\begin{equation}\label{eq:bd}
\mu\frac{\tr(R_j^T B_j)}{\sqrt{2}}>\norm{C_j}_F,
\end{equation}
where $\mu=\frac{K-1}{K+1}$, and $K$ is the desired bound on the condition number of $A_j$. (We used $K=15$ for the planar mappings and $K=5$ for the mesh parameterization example.)  When working with planar meshes, that is $\K_0\subset\Real^2$ we initialized $R_j=I$, and for the surface mesh we picked an arbitrary frame in each 3D triangles so $R_j$ was initialized as an arbitrary $2\times 2$ rotation matrix. We added the boundary conditions as required in Theorems \ref{thm:bijectivity1} or \ref{thm:bijectivity2}, both are sets of linear constraints with the unknowns $\set{\u_j}$ and solve a feasibility problem to get an initial mapping. That is, we added to the left hand side in (\ref{eq:bd}) a new auxiliary variable $t$ (the same $t$ for all faces $f_j$) and minimized $t$. This is a convex problem. When reached a minimum, if $t<0$ we have found a feasible solution and hence a bijective mapping according to Theorem \ref{thm:bijectivity1} or Theorem \ref{thm:bijectivity2}. If the minimal $t$ was greater than zero we reset the rotation $R_j$ to be the rotation closest to $B_j$ and repeated. As explained in \cite{Lipman:2012:BDM} this procedure takes the largest symmetric convex space around the current map and hence allows further reduction of the functional, in this case $t$.

In applications one is often interested in mappings that posses some regularity. In this paper we chose to optimize a standard well-known regularity functional, namely the Dirichlet energy. The Dirichlet energy is defined over triangular meshes as
$E_{\textrm{dir}}(\Phi)=\sum_{f_j\in\K_2}\norm{A_j}_F^2\area{f_j}$, where $\area{f_j}$ is the area of $f_j$. We optimized the Dirichlet energy with the convex constraints (\ref{eq:bd}) using the rotation $R_j$ achieved in the feasibility phase described above. Once converged we reset the rotations and repeat until convergence (usually no more than 3-5 iteration are required). We demonstrate this algorithm in two scenarios: planar mesh mapping and surface mesh parameterization.  The algorithm is implemented in Matlab environment using the MOSEK \cite{MOSEK} and YALMIP \cite{YALMIP} optimization packages.

\begin{figure}[t!] 
\centering
\begin{tabular}{@{\hspace{0.0cm}}c@{\hspace{0.0cm}}c@{\hspace{0.0cm}}c@{\hspace{0.0cm}}c@{\hspace{0.0cm}}c@{\hspace{0.0cm}}}
   \includegraphics[width=0.2\columnwidth]{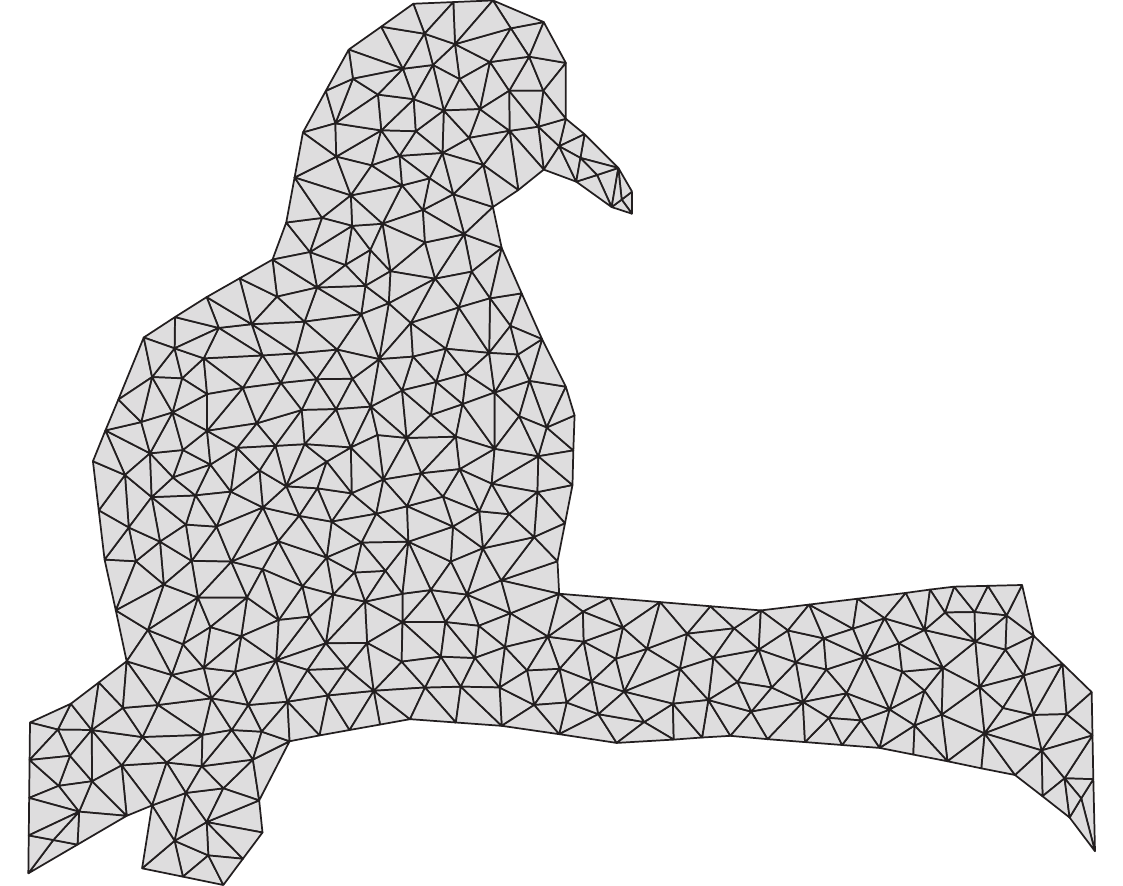}&
  \includegraphics[width=0.2\columnwidth]{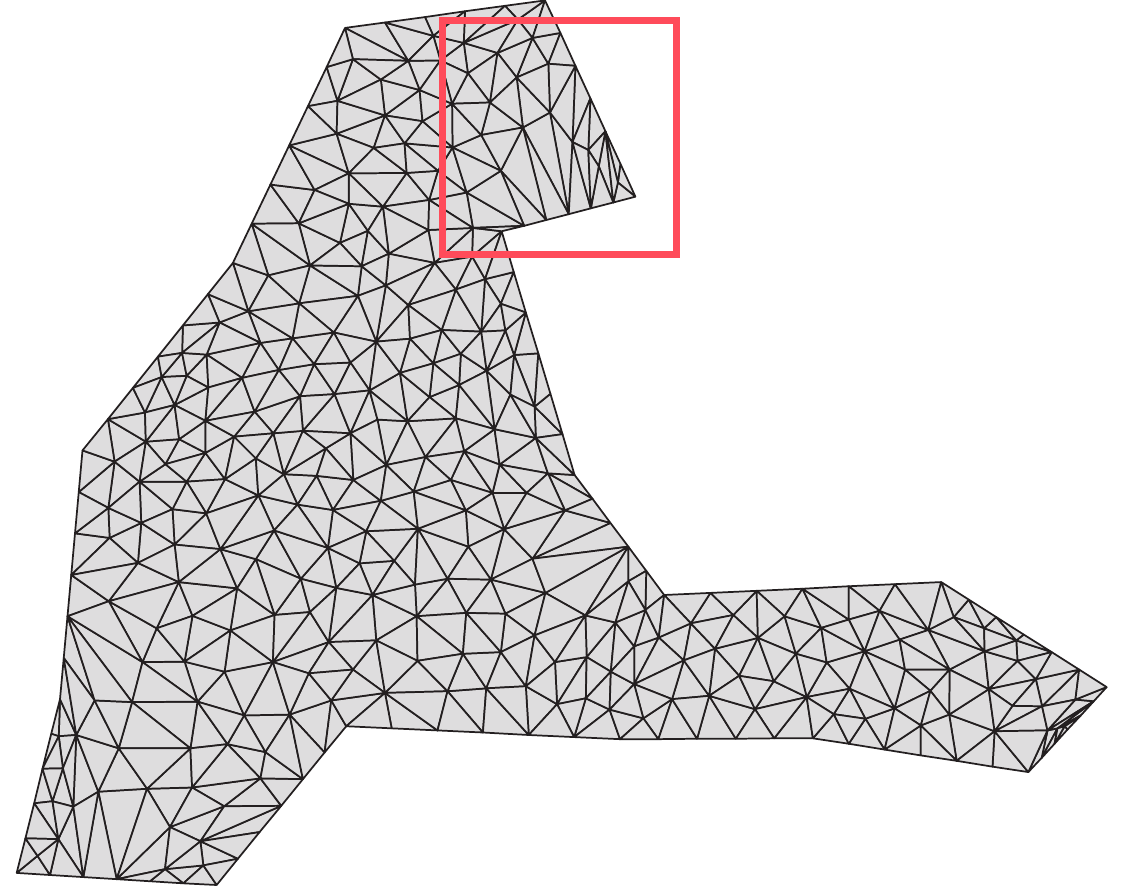}&
   \includegraphics[width=0.2\columnwidth]{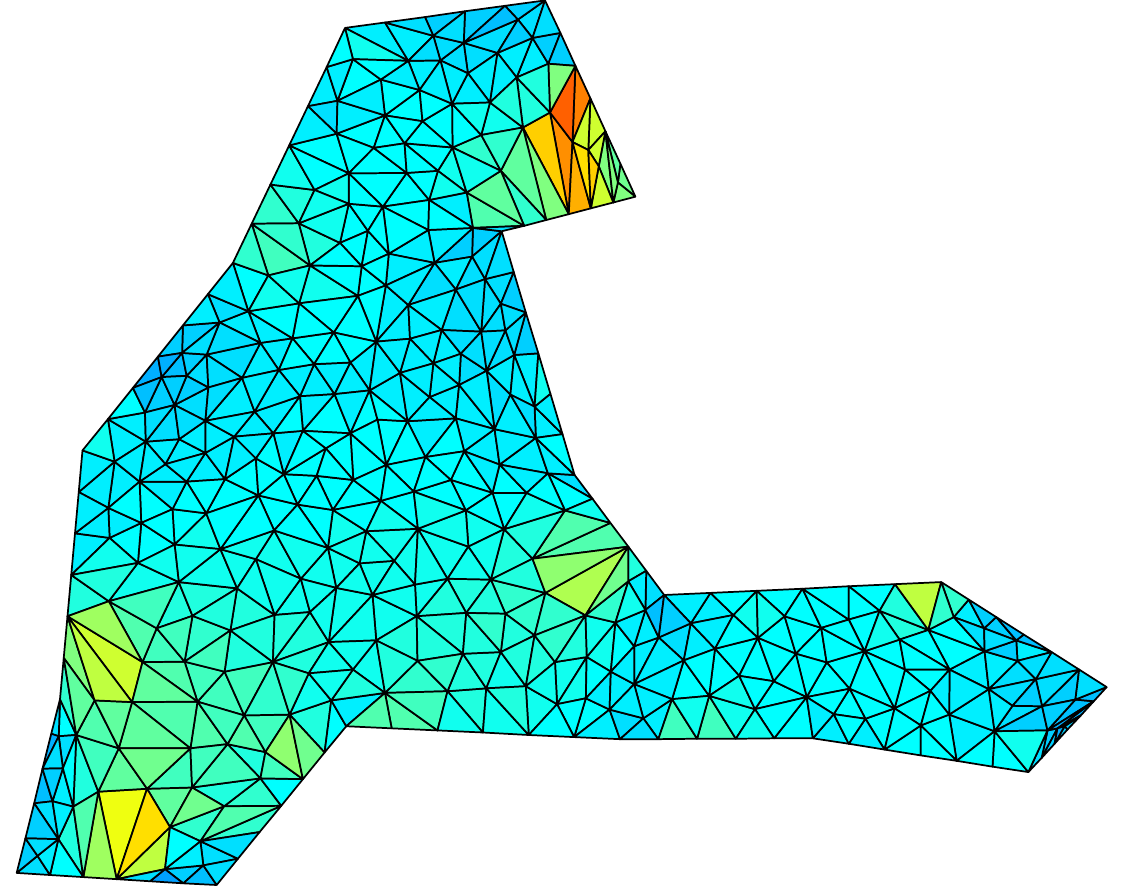}&
    \includegraphics[width=0.2\columnwidth]{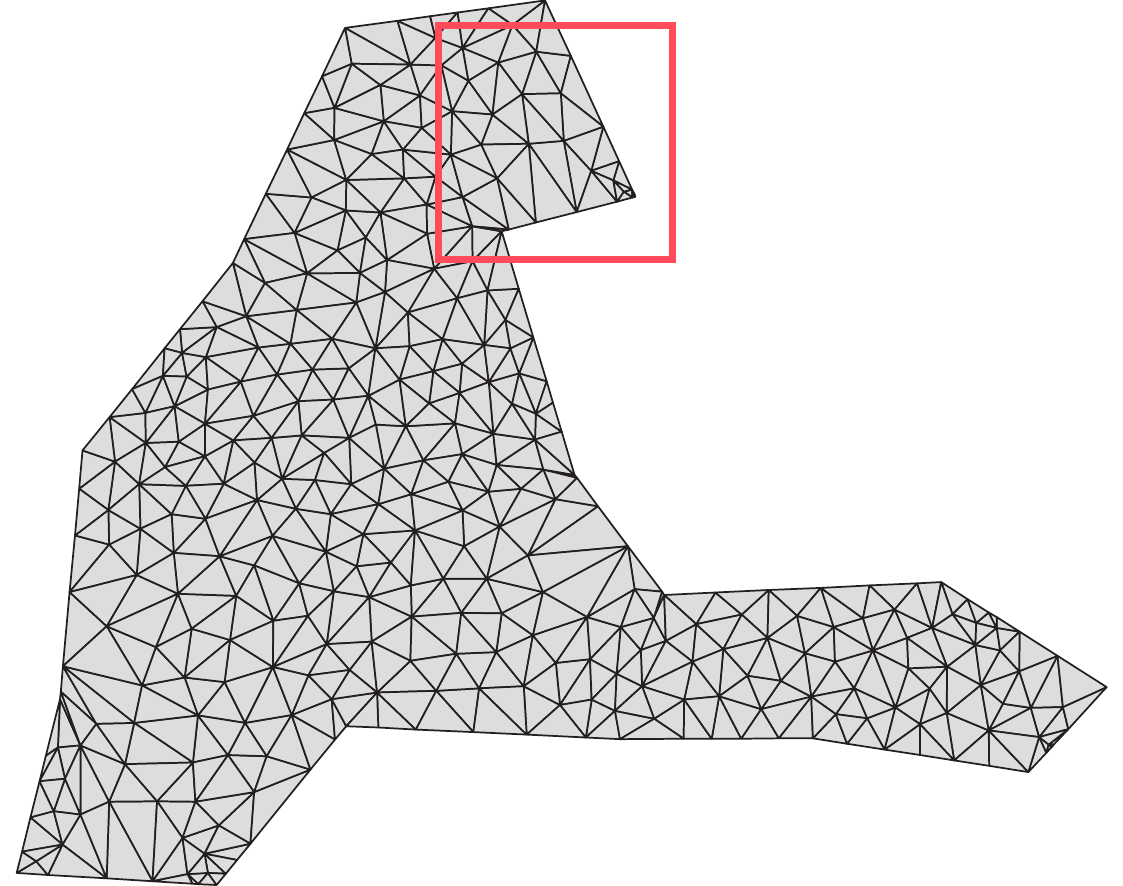}&
    \includegraphics[width=0.2\columnwidth]{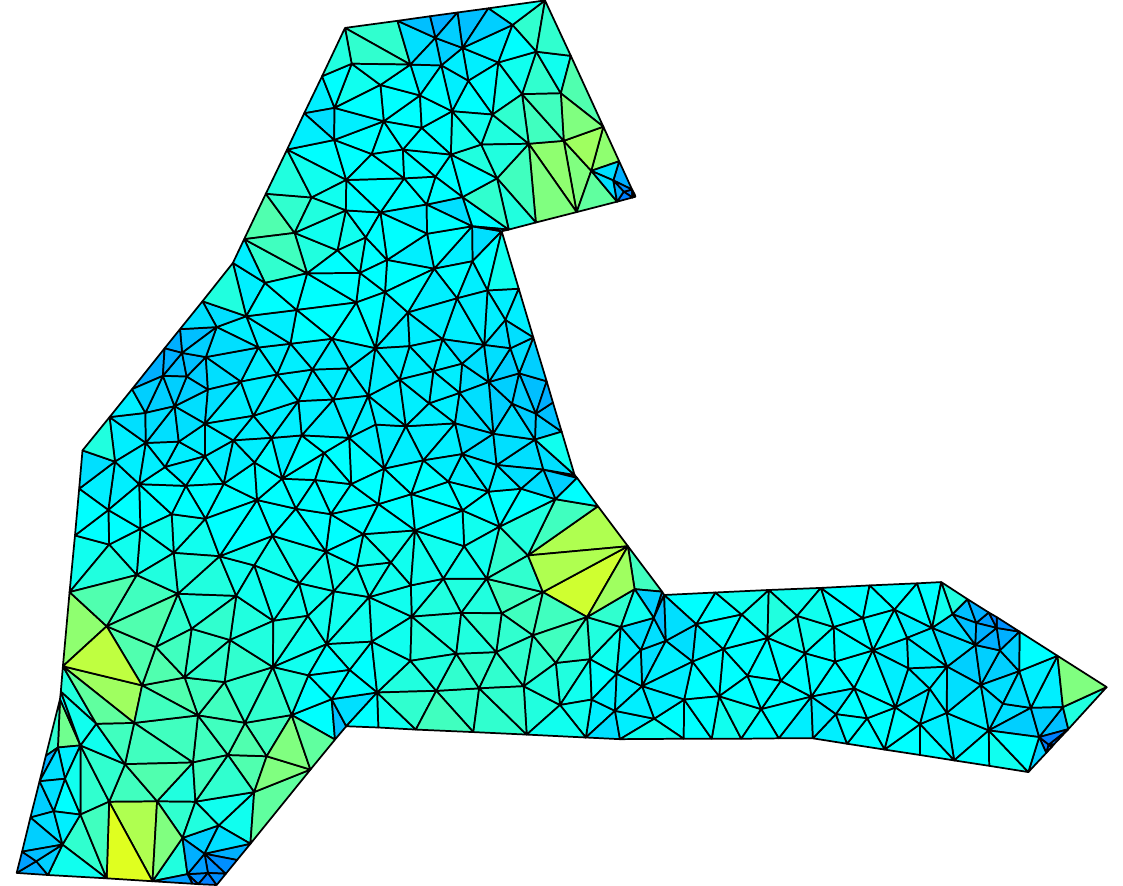}\\
    & $E_{\mathrm{dir}}=0.729$ & & $E_{\mathrm{dir}}=0.692$ &  \\

  \includegraphics[width=0.2\columnwidth]{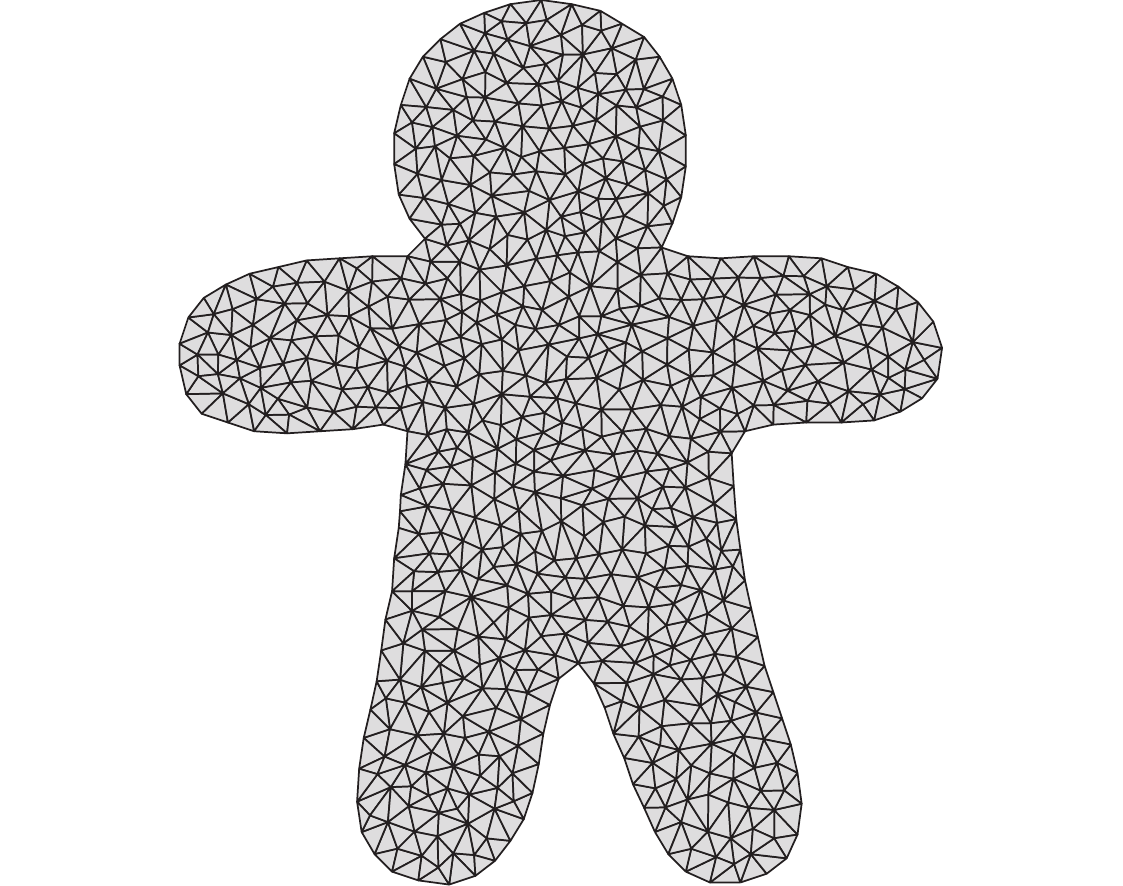}&
  \includegraphics[width=0.2\columnwidth]{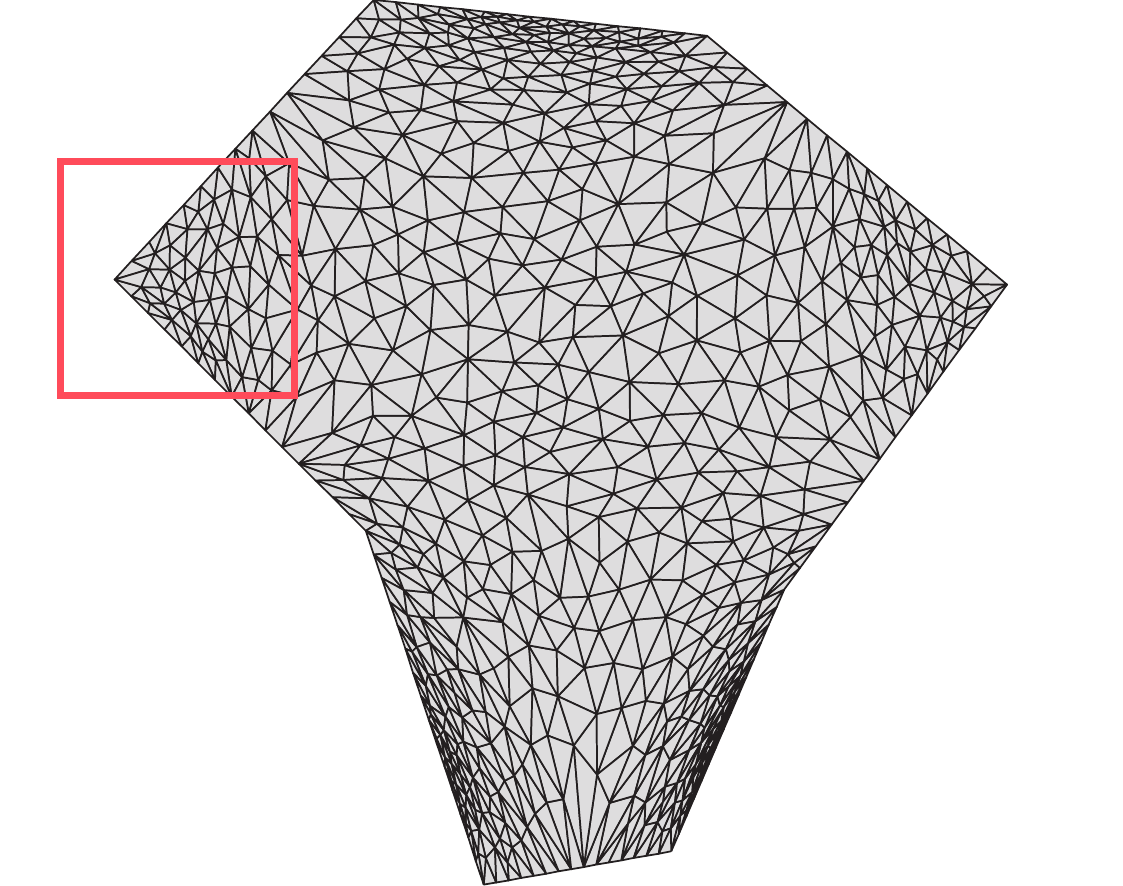}&
   \includegraphics[width=0.2\columnwidth]{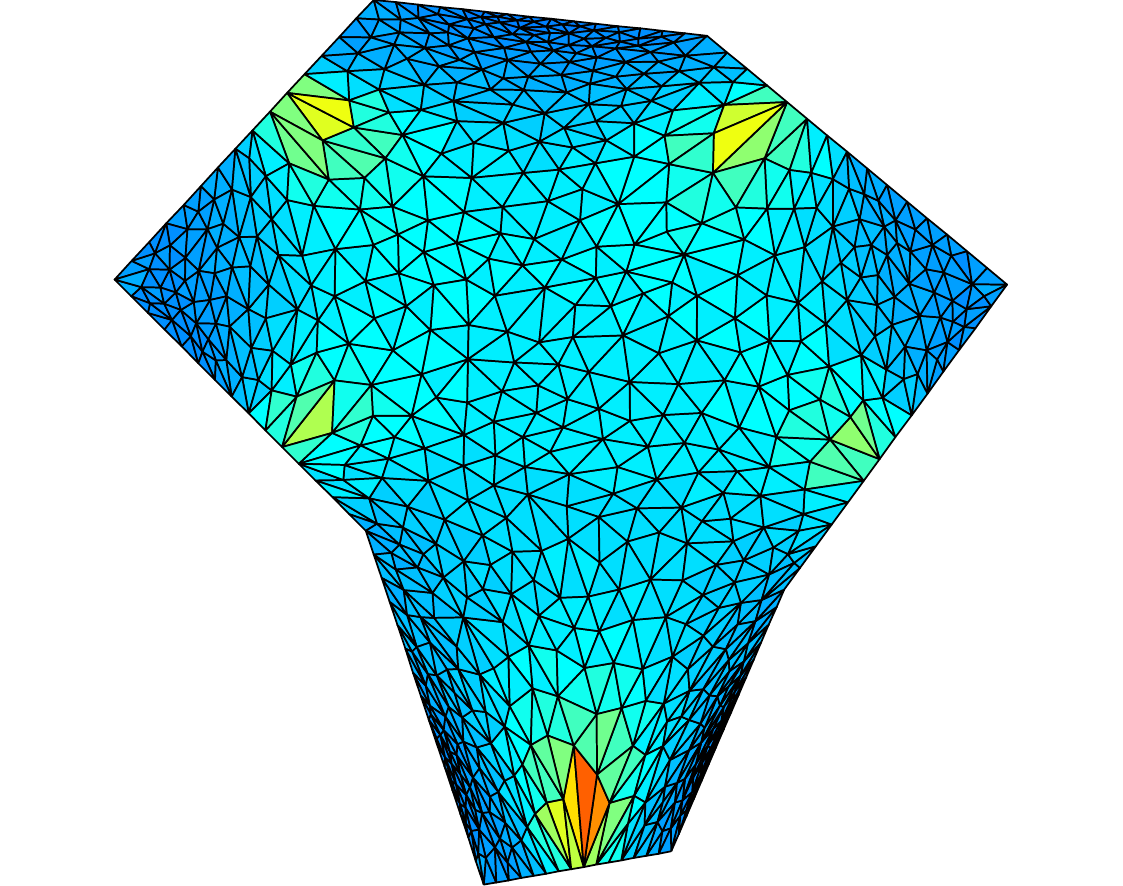}&
    \includegraphics[width=0.2\columnwidth]{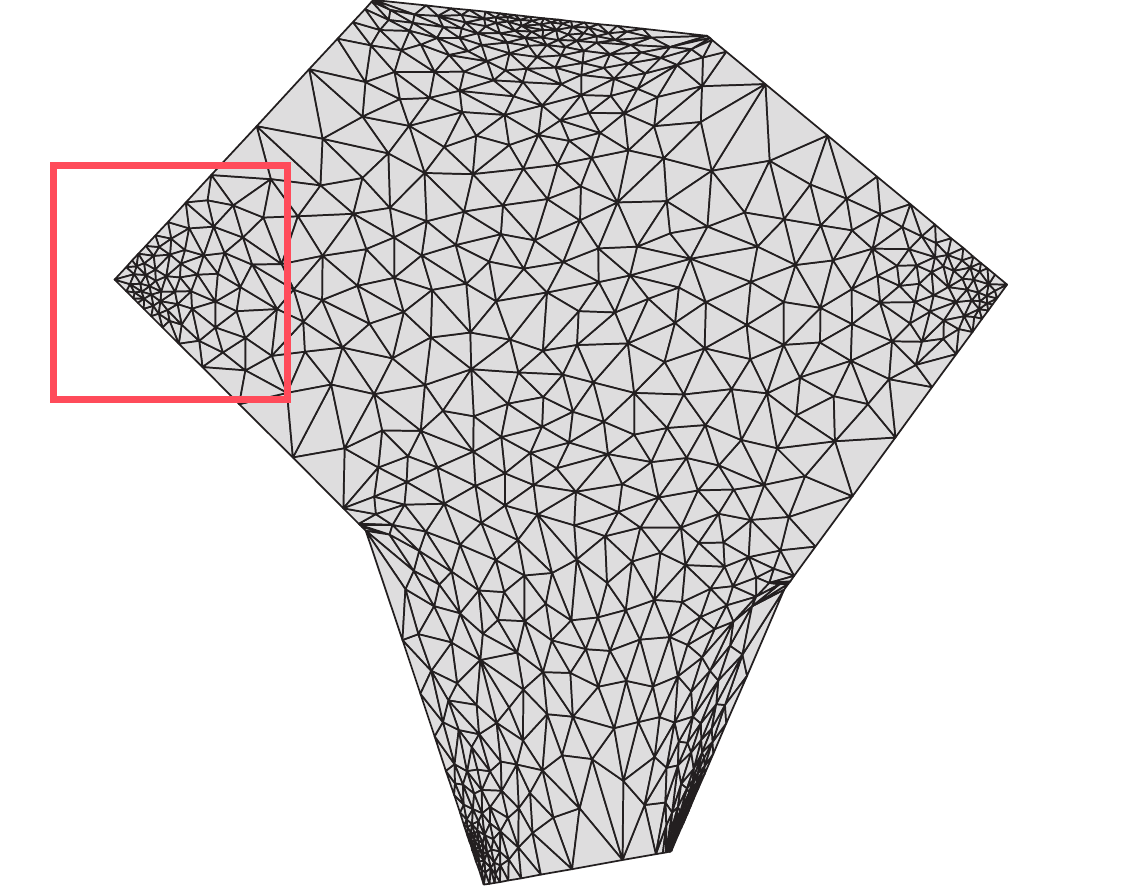}&
    \includegraphics[width=0.2\columnwidth]{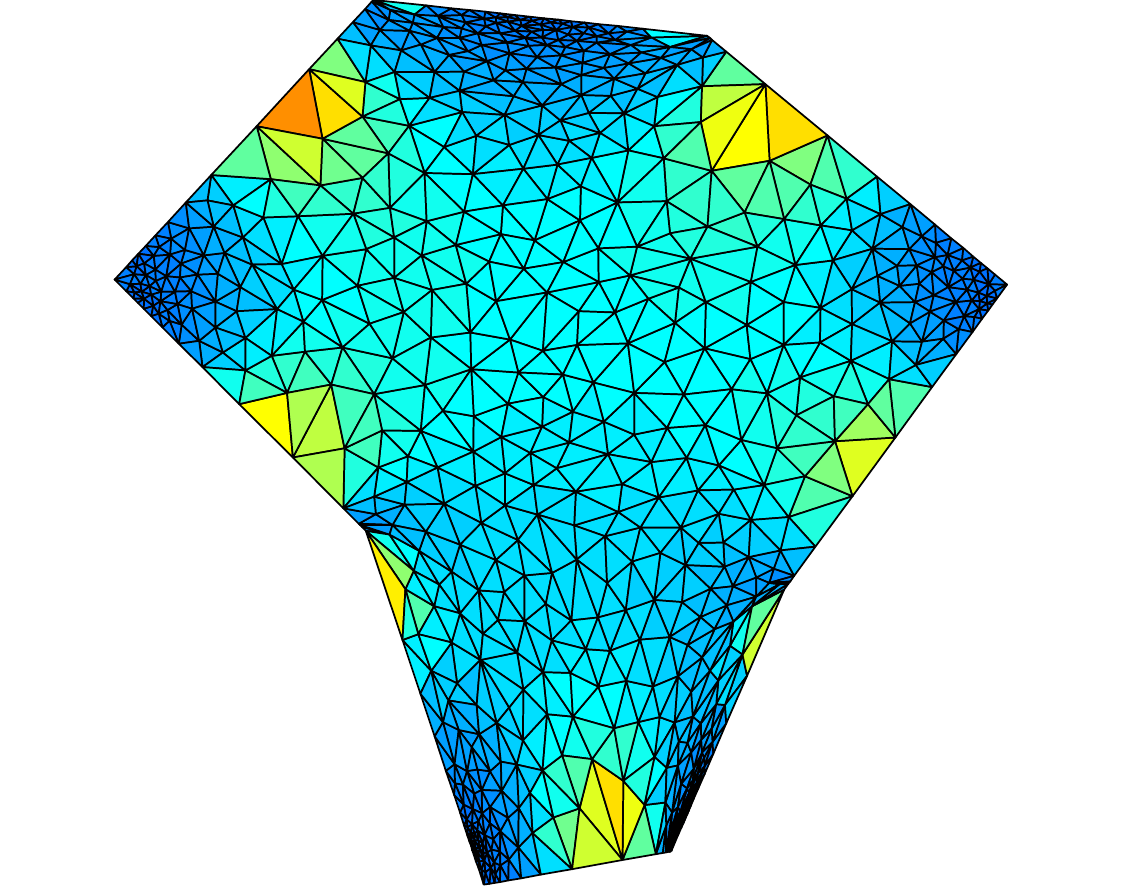}\\
    & $E_{\mathrm{dir}}= 1.065$ & & $E_{\mathrm{dir}}=0.927$  & \\

  \includegraphics[width=0.2\columnwidth]{figures/examples/woody.pdf}&
  \includegraphics[width=0.2\columnwidth]{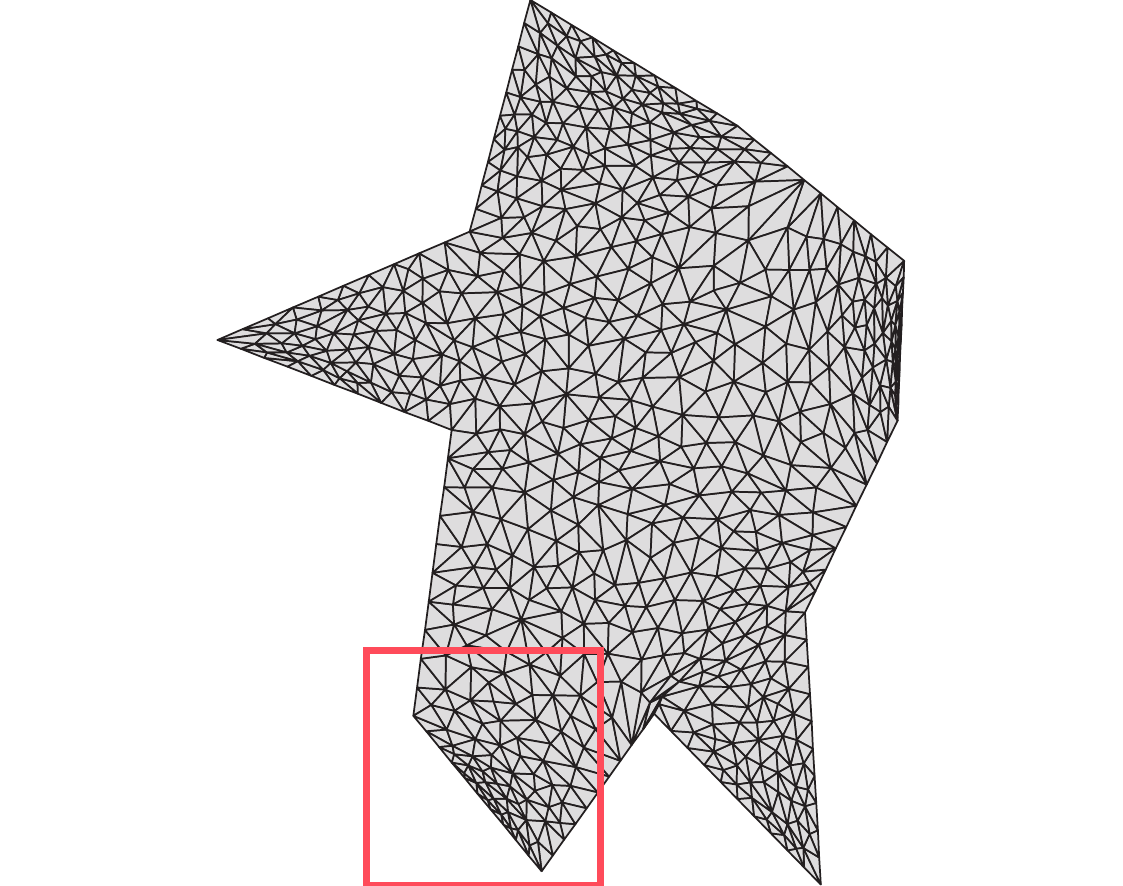}&
   \includegraphics[width=0.2\columnwidth]{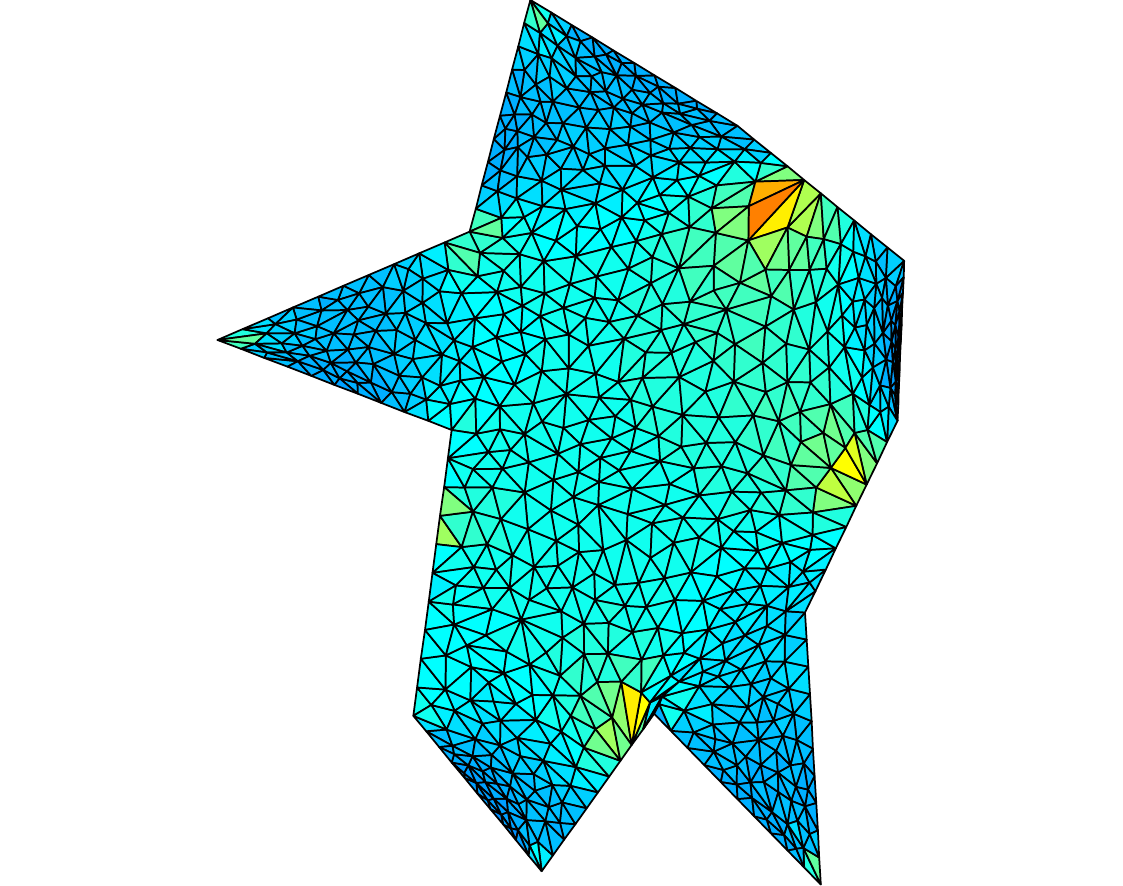}&
    \includegraphics[width=0.2\columnwidth]{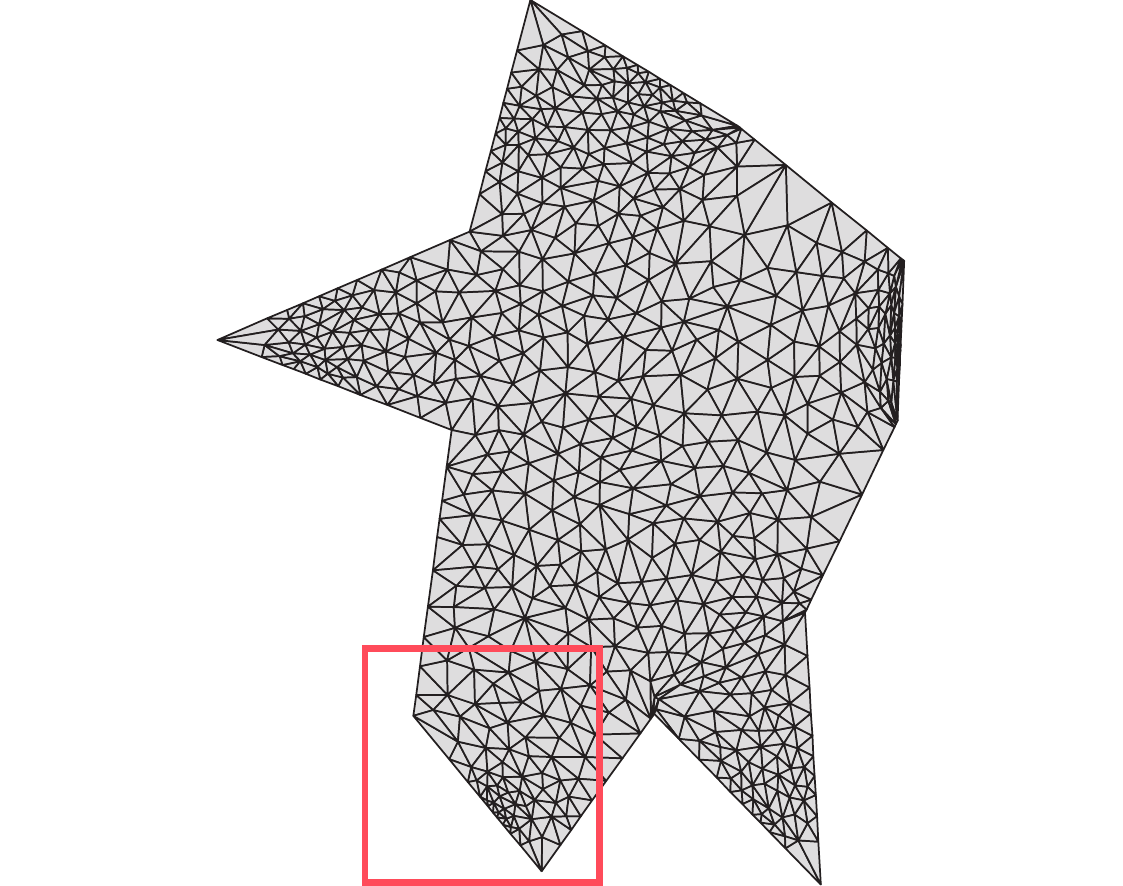}&
    \includegraphics[width=0.2\columnwidth]{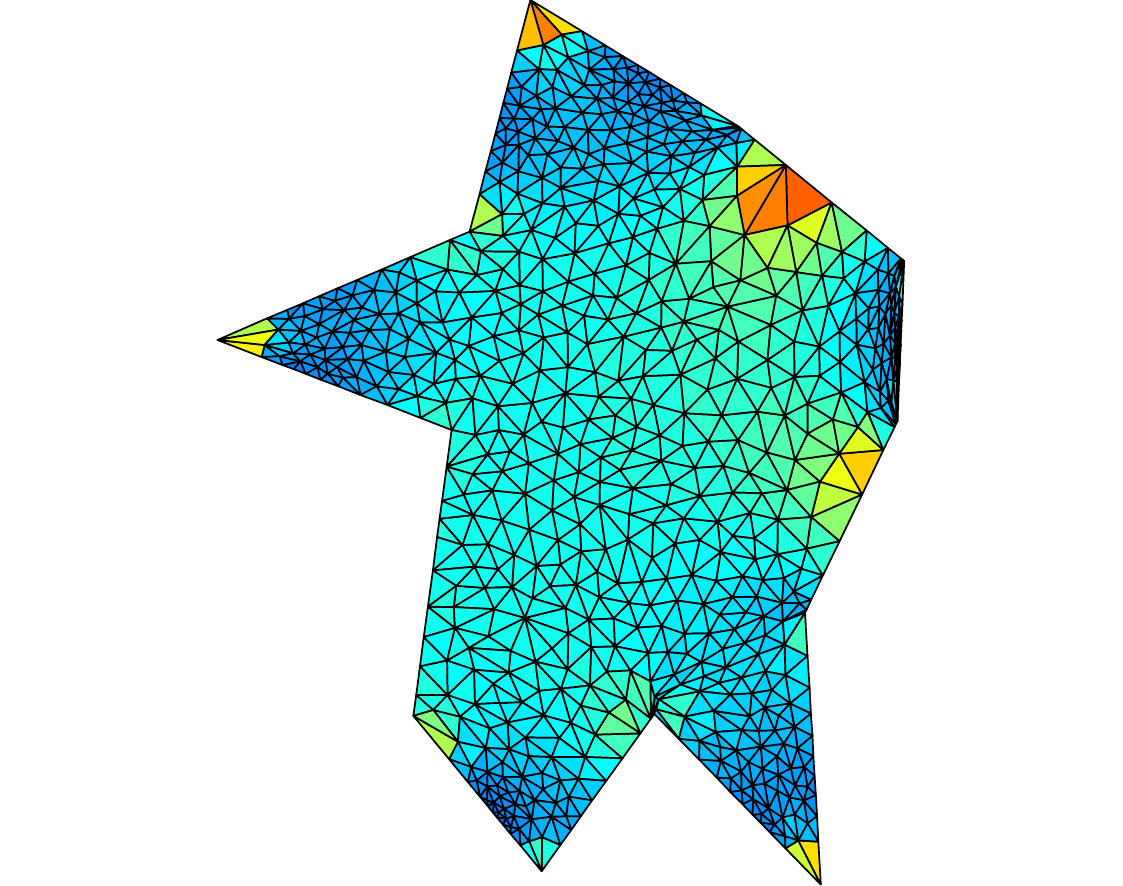}\\
    & $E_{\mathrm{dir}}=0.840$ & & $E_{\mathrm{dir}}=0.797$ & \\

  \includegraphics[width=0.2\columnwidth]{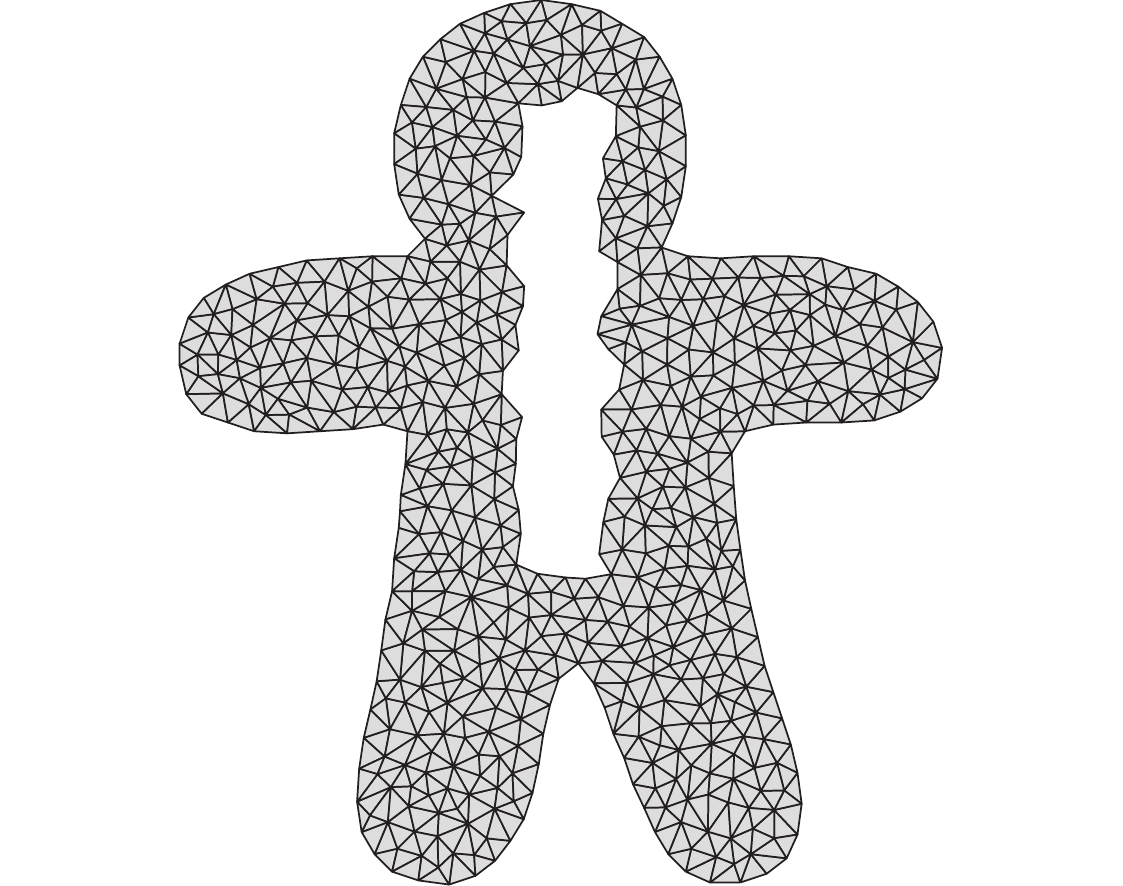}&
    \includegraphics[width=0.2\columnwidth]{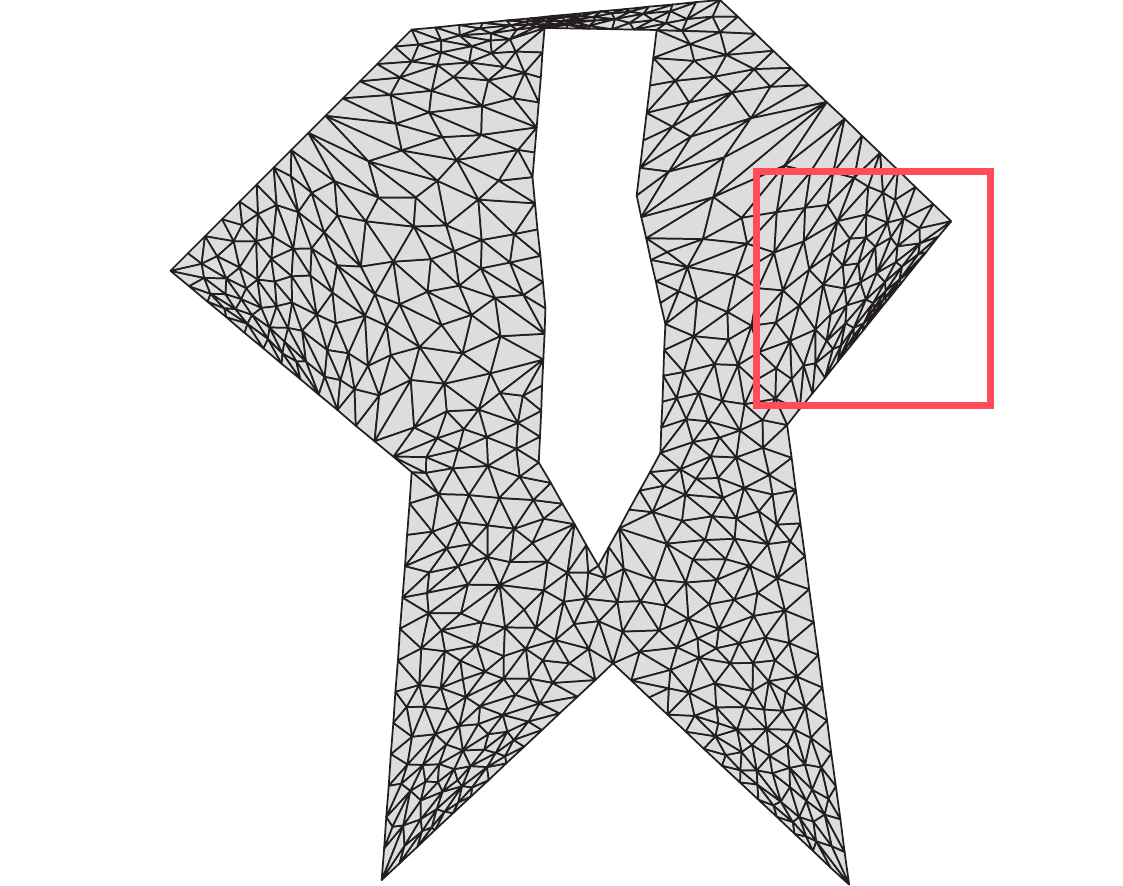}&
  \includegraphics[width=0.2\columnwidth]{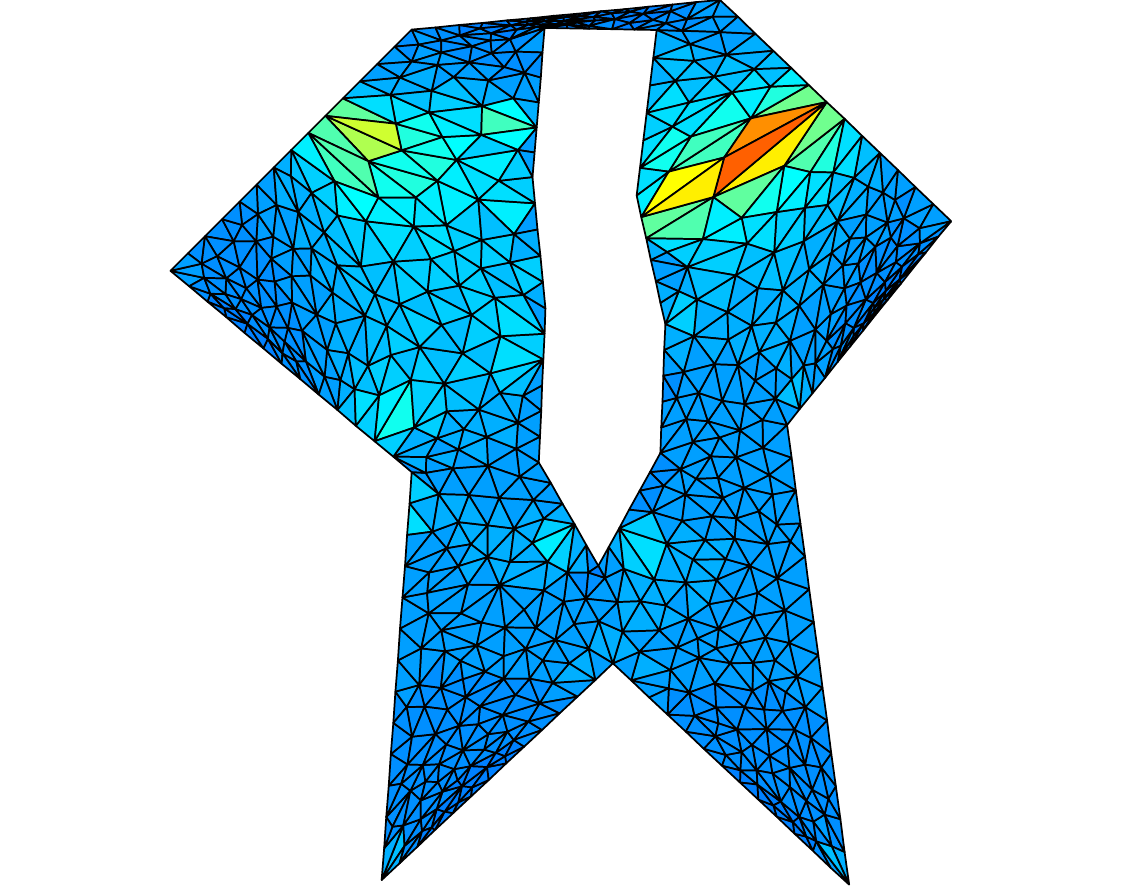}&
    \includegraphics[width=0.2\columnwidth]{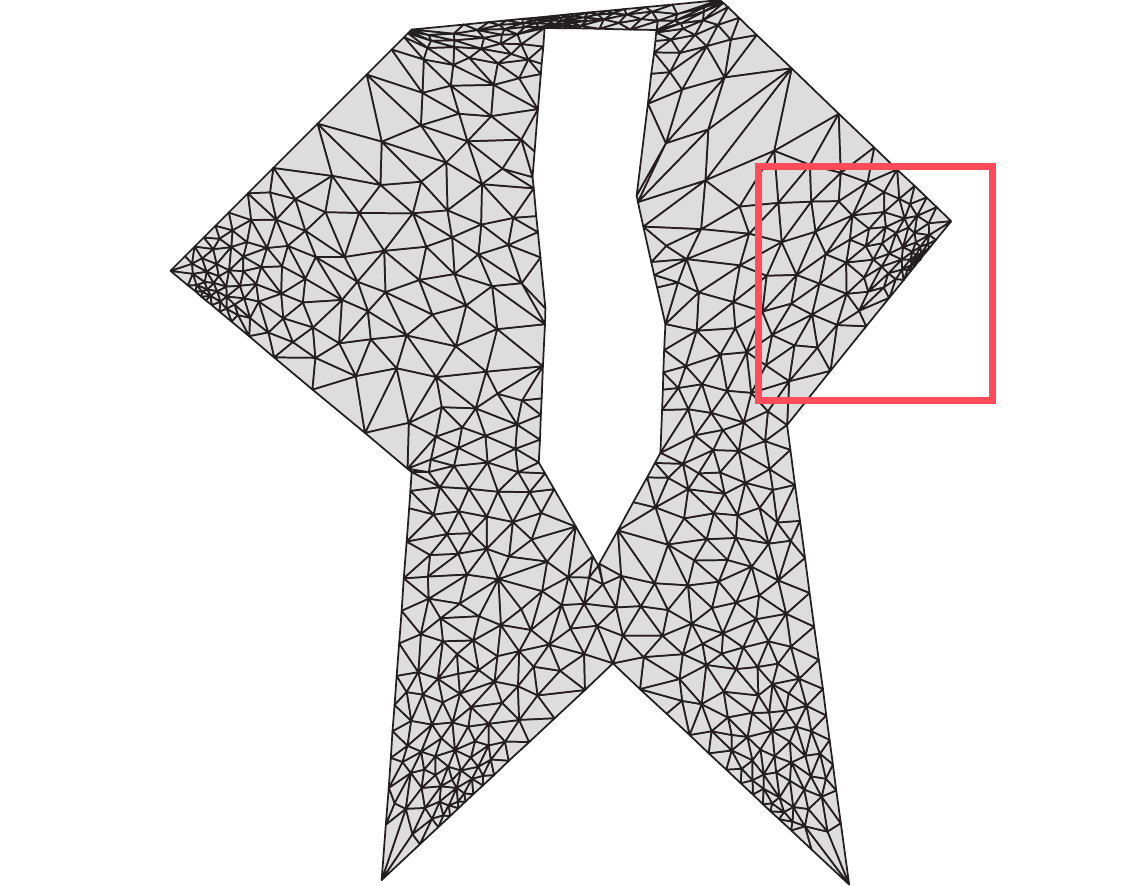}&
    \includegraphics[width=0.2\columnwidth]{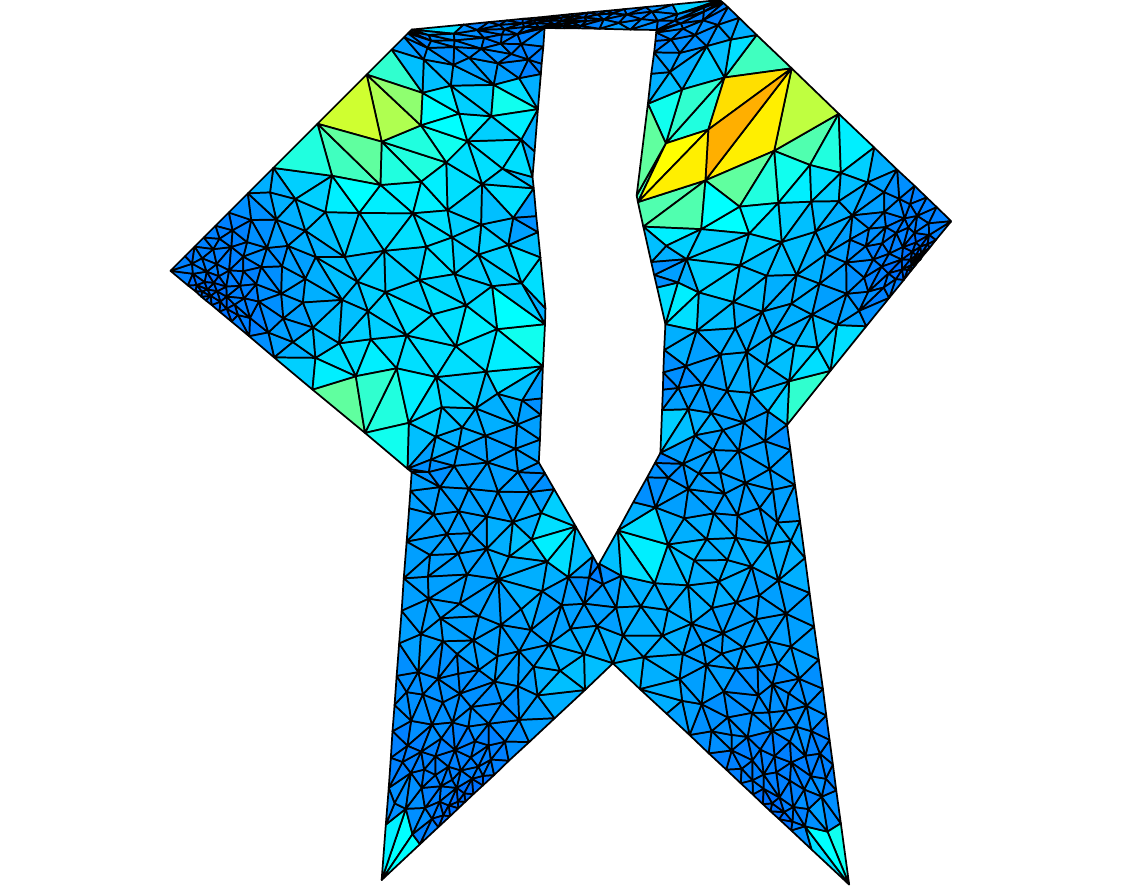}\\
    & $E_{\mathrm{dir}}=0.801$ & & $E_{\mathrm{dir}}=0.726$ &  \\

   \end{tabular}
  \caption{  Computing bijective simplicial mappings of planar meshes (source meshes are shown in left column) onto a polygonal domains using: Dirichlet energy with orientation preservation constraints and a \emph{fixed} uniform boundary map (second and third columns); and Dirichlet energy with orientation preservation constraints and the linear boundary conditions of Theorem {\ref{thm:bijectivity2}} (fourth and fifth columns). Note how the boundary map is optimized in the latter case to reduce the Dirichlet energy of the map (see $E_{\textrm{dir}}$ value indicating the total Dirichlet energy). The third and fifth columns show the norm of the gradient of each map color coded (red indicates large gradients and blue small gradients).  }\label{fig:woody}
\end{figure}

Figure \ref{fig:woody} depicts four examples of mapping a planar triangular mesh (the left column shows the source meshes) onto a polygonal domain. Each row shows the result of mapping the respective source mesh to a polygonal domain by minimizing the Dirichlet energy of the map while constraining the map to be orientation-preserving  using the conic formulation described above. The second and third columns in Figure \ref{fig:woody} show the result of minimizing the Dirichlet energy with the orientation preservation constraints while prescribing the boundary map to be the \emph{fixed} uniform map, that is, mapping the boundary vertices of the source mesh to equally spaced locations along each edge of the target polygon. Since this boundary map is a bijection, Theorem \ref{thm:bijectivity1} implies that $\Phi$ is a bijection. The Dirichlet energy of the resulting maps ($E_{\textrm{dir}}$) are written below the second column images.  The third column shows the gradient norm of the map (i.e., $\norm{A_j}_F$) color-coded. The fourth and fifth columns show the result of minimizing the Dirichlet energy with the orientation preservation constraints but this time with the boundary conditions of Theorem \ref{thm:bijectivity2} (i.e., Eq.~(\ref{e:linear_con})). Theorem \ref{thm:bijectivity2} guarantees the resulting map to cover the target domain exactly, and to be injective over the interior of the source mesh. Note how the boundary map is optimized (several areas of interest are highlighted with red squares) to reduce the overall Dirichlet energy of the resulting map (compare the energies $E_{\textrm{dir}}$ at the bottom of each image in the second and fourth columns). The fifth column, similarly to the third column shows the magnitudes of the gradients.

\begin{figure}[t] 
\centering
\begin{tabular}{cc}
  \includegraphics[width=0.25\columnwidth]{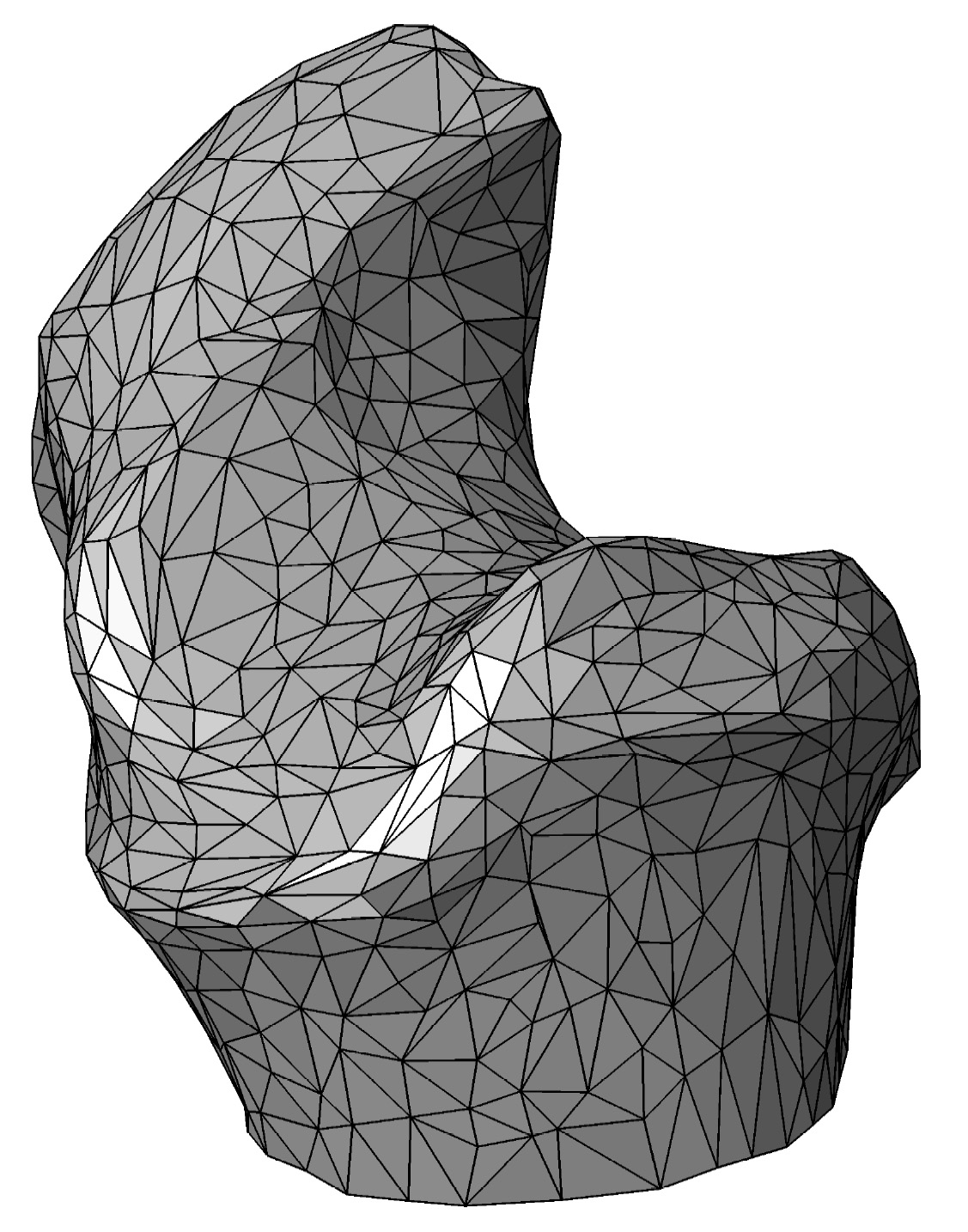}&
  \includegraphics[width=0.25\columnwidth]{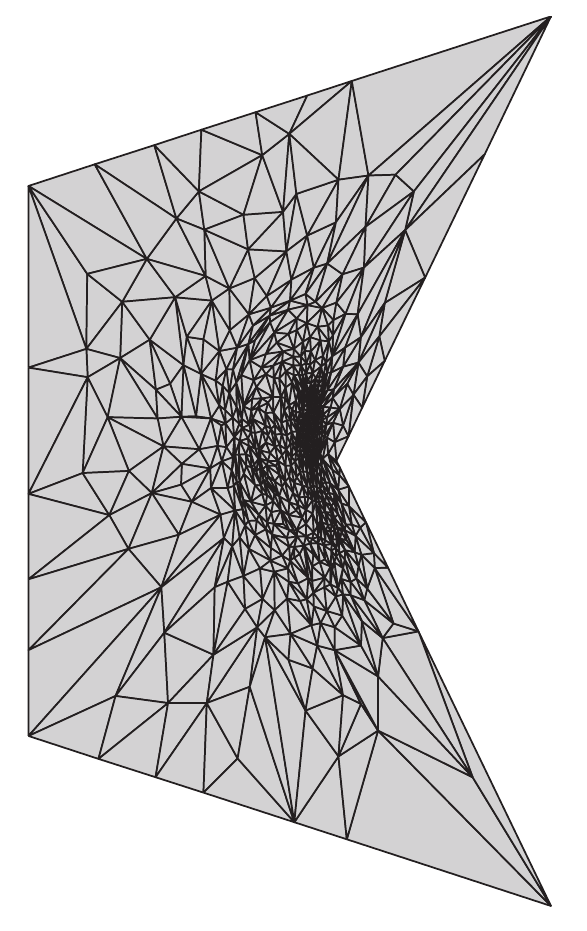}\\
  (a) & (b) \\

   \end{tabular}
  \caption{ Injectively parameterizing a surface mesh (a) onto a polygonal domain in the plane (b) by optimizing the Dirichlet energy with orientation preservation constraints and the linear boundary conditions of Theorem \ref{thm:bijectivity2}.}\label{fig:param}
\end{figure}

Figure \ref{fig:param} shows an example of mapping a triangular mesh embedded in 3D (mesh surface, in (a)) onto a polygonal domain in the plane in (b). Here as-well we have optimized the Dirichlet energy, together with the orientation preserving constraints as above. The linear boundary conditions of Theorem \ref{thm:bijectivity2} allowed the optimization to choose the boundary map that allowed low Dirichlet energy of this parameterization under the given assignment.


\section{Conclusions}
This paper utilizes the concept of degree of simplicial maps over cycles to formulate and prove three sufficient conditions for guaranteeing that a simplcial map of a mesh with boundary into a polytope is a bijection. The conditions are practical in the sense that they can be incorporated into algorithms that preserve orientations of simplices to produce bijective mappings. The conditions are appropriate for cases where a bijective mapping of a mesh with boundary onto a known polytope is sought and the boundary mapping is either supplied or unknown and needs to be optimized as-well.

A limitation of the sufficient conditions developed in this paper is that they still do not allow complete freedom in optimizing the boundary map. For example, they do not allow boundary faces of the mesh to slide from one boundary face of the target polytope to another boundary face of the polytope. Ideally, such a property would allow optimizing the boundary map so to reduce the energy further.

Finally, the paper focuses on the case of meshes with boundary and a natural question is how to bijectively map manifold meshes without boundary. This question involves topological questions and we leave this very interesting problem to future work.

\newpage
\section{Appendix}

\parpic[r]{
\begin{minipage}{0.15\columnwidth}\hspace{-0.3cm}
    \includegraphics[width=0.75\textwidth]{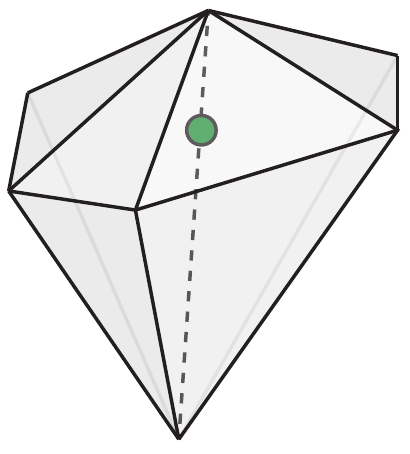}
\end{minipage}}

In the appendix we prove the open map property of non-degenerate orientation preserving simplicial maps.
In the proof we will use the notion of \emph{1-ring} which is defined for arbitrary $x\in\interior{\M}$ as $R_x=\bigcup_{\sigma\in\K_d, x\in \sigma} \sigma$.
The inset depicts an example of the 1-ring of a point (green) in the relative interior of an edge in a tetrahedral mesh ($d=3$).\\

\textbf{Lemma \ref{lem:open_map}}
\textit{
Let $\Phi:\M\subset \Real^n\too \Real^d$, $n\geq d$, be a non-degenerate orientation preserving simplicial map of a compact $d$-dimensional mesh $\M$ into $\Real^d$. Then $\Phi$ is an open map.}
\begin{proof}
We need to show that for every open set $U\subset\interior{\M}$ the set $\Phi(U)$ is an open set. That is, any  point $x\in U$ is mapped to an interior point $\Phi(x)\in\interior{\Phi(U)}$. We will use the 1-ring of $x$, $R_x=\bigcup_{\sigma\in\K_d, x\in \sigma} \sigma$ . Let $W=U \cap \interior{R_x}$. Since $\Phi$ is non-degenerate, $\Phi(x)$ has some positive distance to $Q:=\Phi(\partial \closure{W})$. Let us consider a small open neighborhood $V$ of $\Phi(x)$ such that all points in $V$ are strictly closer to $\Phi(x)$ than $Q$.\\
We will compute the Brouwer degree (see F\"{u}hrer characterization on page 39 in \cite{outerelo2009mapping}) $\deg(\Phi,W,p)$, for points $p\in V$. Let $Y=\cup_{\tau\in\K_{d-1}}\Phi(\tau)$. $Y$ has measure zero. Since $\Phi$ is non-degenerate and continuous we can find a point $q=\Phi(x')\in V\setminus Y$, where $x'\in W$. As $\Phi$ is orientation preserving we have $\deg(\Phi,W,q)\geq 1$.  By homotopy invariance $\deg(\Phi,W,p)$ is constant for all $p\in V$. Hence, $\deg(\Phi,W,p)\geq 1$ for all $p\in V$. The existence of solution property of the degree now implies that any such $p\in V$ has a pre-image in $W$, $p=\Phi(z)$, $z\in W\subset U$. Hence we proved $\Phi(x)\in V\subset \Phi(U)$, as-required.\\

\end{proof}

\bibliographystyle{amsplain}
\bibliography{bijective}

\end{document}